\documentclass[11pt,a4paper]{article}
\usepackage{threeparttable}
\usepackage{stmaryrd}
\usepackage{amssymb}
\usepackage{mathrsfs}
\usepackage{latexsym}
\usepackage{epsfig}
\usepackage{subfigure}
\usepackage{graphicx}
\usepackage{epstopdf}
\usepackage{amsthm}
\usepackage{amsfonts}
\usepackage{amsmath}
\usepackage{amssymb,amsmath,amsthm,amscd}
\usepackage{algorithm}
\usepackage{algorithmic}
\usepackage{multirow}
\usepackage{xcolor}
\usepackage{array}
\usepackage{booktabs}
\usepackage{bm}
\usepackage[numbers,sort&compress]{natbib}
\usepackage{lscape}
\usepackage[indentafter]{titlesec}
\usepackage[colorlinks,linkcolor=red,linkcolor=blue,linktocpage]{hyperref}
\usepackage{threeparttable}
\usepackage{caption}
\usepackage{enumitem,calc}
\newcommand\preitem{\mdseries\textbullet\space}
\newlist{desclist}{description}{3}
\setlist[desclist,1]{format=\preitem\bfseries,leftmargin=\widthof{\preitem},style=sameline}
 \makeatletter

\newcommand{\Rmnum}[1]{\expandafter\@slowromancap\romannumeral #1@}
\makeatother \providecommand{\U}[1]{\protect \rule{.1in}{.1in}}
\floatname{algorithm}{Method}
\newtheorem{theorem}{Theorem}[section]

\newtheorem{lemma}{Lemma}[section]

\makeatother \providecommand{\U}[1]{\protect \rule{.1in}{.1in}}
\hypersetup{colorlinks=true, citecolor=blue}
\DeclareMathOperator{\sech}{sech}

\def \W#1{\widehat{#1}}
\def \DW#1#2{|\widehat{#1};\widehat{#2}|}

\title{New dynamics of the classical and nonlocal Gross-Pitaevskii equation  with a parabolic potential}

\author{Shi-min Liu,~~Hua Wu,~~  Da-jun Zhang\footnote{Corresponding author. Email: djzhang@staff.shu.edu.cn}\\
{\small\it Department of Mathematics,
 Shanghai University, Shanghai 200444,  P.R. China}}

\date{\today}

\begin{document}

\maketitle

\begin{abstract}
Solutions of the classical and nonlocal  Gross-Pitaevskii (GP) equation with a parabolic potential and a gain term
are derived by using a second order nonisospectral Ablowitz-Kaup-Newell-Segur system and
reduction technique of double Wronskians.
Solutions of the classical GP equation show typical space-time localized characteristics.
An interesting dynamics, solitons carrying an oscillating wave, are found with mathematical analysis and illustrations.
Solutions of some nonlocal cases are also illustrated.

\vskip 6pt

\end{abstract}

%

\section{Introduction}
The well-known  Gross-Pitaevskii (GP) equation is given by \cite{Gross-1961, Gross-1963, Pitaevskii-1961}
\begin{equation}\label{GP}
\mathrm{i} \hbar \frac{\partial \psi}{\partial t}=\left(-\frac{\hbar^{2}}{2 m} \nabla^{2}+V_{ext}+g|\psi|^{2}\right) \psi,
\end{equation}
where $\psi=\psi(\mathbf{x},t)$ is the wave function with $\mathbf{x}$ the three-dimensional spatial coordinate,
$i$ is the imaginary unit,
$\hbar$ is the Planck constant, $m$ is the mass of the boson, $V_{ext}$ is the external potential
and $g=4\pi\hbar^2/a_s$ is a parameter that measures the atomic interactions with $a_s$ the scattering length of two interacting bosons.
The GP equation can be used to describe the behaviour of the wave function in Bose-Einstein condensates (BECs),
it has the same mathematical form as the nonlinear Schr\"odinger equation (NLS) with  an external potential.
For more details about BECs and the GP equation, one can refer to Ref.\cite{Liu-2019}.
This equation was first derived by Gross \cite{Gross-1961} and Pitaevskii \cite{Pitaevskii-1961} and bears by their name.
Since then, there are many works from different sides on study of the GP equation with various forms of external potential, such as harmonic potential
\cite{Brazhnyi-2003,Serkin-2007,Liang-2005,Zhang-2009}, Gaussian trap \cite{Tempere-2001,Delande-2014} and
optical lattice potential \cite{Merhasin-2006,Sakaguchi-2009,Dong-2013}.

In this paper,  we investigate the GP equation in (1+1)-dimension with a parabolic potential and a gain term \cite{Brazhnyi-2003,Serkin-2007}
\begin{equation}\label{GP2}
iq_t+q_{xx}+2|q|^2q+(\delta x^2+i\alpha)q=0,
\end{equation}
where $\delta$ is a real constant, the gain $\alpha=\alpha(t)$ is a real function of $t$ governed by
\begin{equation}\label{a}
\alpha_t+2\alpha^2=2\delta.
\end{equation}
The GP equation of this form was first investigated by Gupta in 1979 \cite{Gupta-1979},
where he gave a transformation to convert this equation to an integrable nonisospectral NLS.
Guaranteed by such a relation, the equation \eqref{GP2} has been solved by using different methods of integrable systems,
including Inverse Scattering Transform \cite{Gupta-1979,Khawaja-U-A-2009,Zhang-Y-J-2017},
Darboux transformation \cite{Vinoj-2001,Zhang-Y-J-2014}, Wronskian technique \cite{Sun-2014},
and so on.

The purpose of this paper is to investigate both classical and nonlocal form of the GP equation \eqref{GP2}
through the bilinear method and provide solutions in double Wronskian form.
As a first nonlocal integrable system, the nonlocal NLS equation was introduced by
Ablowitz and Musslimani \cite{AM-PRL-2013} in 2013 as a PT-symmetric model.
From then on nonlocal integrable systems received intensive attention from many aspects
(e.g. \cite{AM-Nonl-2016,Ablowitz-SAPM-2016,Fokas-2016,YanY-arxiv-2017,Cau-SAPM-2018,ChenDLZ-SPAM-2018,Zhou-SAPM-,AM-JPA-2019}),
but so far there is no report on nonlocal nonisospectral integrable systems, which contains $x$-dependent coefficients.
In this paper, we will show that the nonisospectral NLS equation related to the GP equation \eqref{GP2} allows a nonlocal form.
This fact enables us to investigate the nonlocal form of equation \eqref{GP2}.
We will employ the reduction approach on double Wronskians we proposed recently \cite{ChenDLZ-SPAM-2018}
to get solutions of the nonlocal nonisospectral NLS equation as well as the nonlocal GP equation.

The paper is organized as follows.
In Sec.\ref{sec-2} we  derive the classical and nonlocal GP equations
from a coupled nonisospectral system in the Ablowitz-Kaup-Newell-Segur (AKNS) hierarchy.
We also give a  bilinear form of \eqref{GP2} and obtain its $N$-soliton
solutions ($N$SS) through the Hirota method.
In Sec.\ref{sec-3}  we implement reduction technique on double Wronskians
and obtain  solutions in double Wronskian form for both classical and nonlocal GP equations.
Then in Sec.\ref{sec-4} typical dynamics of solutions are analyzed and illustrated.
Finally, Sec.\ref{sec-5} serves for conclusions.

\section{Nonlocal GP and nonlocal isospectral NLS equation}\label{sec-2}

To obtain integrable nonlocal GP equation, we start from the ZS-AKNS spectral problem with time evolution \cite{Zakharov-Shabat-1972,AKNS-1973}
\begin{equation}
\label{lax-a}
\Phi _x= \begin{pmatrix}
-\eta & Q \\
R &\eta
\end{pmatrix}\Phi=M(\eta,Q,R)\,\Phi,
\end{equation}
\begin{equation}
\label{lax-b}
\Phi_t= \begin{pmatrix}
A & B\\
C & -A
\end{pmatrix}
\Phi=N(\eta,Q,R)\,\Phi,
\end{equation}
where $\Phi=(\phi_1,\phi_2)^T$, $Q,R$ are functions of $x,t$ and $\eta$ is the spectral parameter depending on $t$.
The compatibility condition $ M_t-N_x+[M,N]=0 $ gives rise to
\begin{equation}
A=\partial ^{-1}(R,Q)\begin{pmatrix}
-B\\
C
\end{pmatrix} - \eta _t x+A_0\,,
\end{equation}
and
\begin{equation}
  \label{eq8}
\begin{pmatrix}
Q \\ R
\end{pmatrix}_t=L \begin{pmatrix}
-B\\
C
\end{pmatrix} -2 \eta \begin{pmatrix}
-B\\
C
\end{pmatrix}-2 A_0\sigma_3 \begin{pmatrix}
Q\\
R
\end{pmatrix}-2 \eta _t \sigma_3\begin{pmatrix}
xQ\\
xR
\end{pmatrix}\,.
\end{equation}
Here $A_0$ is a function of $\eta$ and $t$, but independent of $x$, and
$$L=-\sigma_3 \partial + 2 \begin{pmatrix}
Q\\
-R
\end{pmatrix} \partial ^{-1}(R,Q),\quad \sigma_3=\begin{pmatrix}
1 & 0\\
0 & -1
\end{pmatrix},\quad \partial=\frac{\partial}{\partial x},
\quad \partial^{-1}\ =\frac{1}{2}\left(\int^{x}_{-\infty}-\int^{+\infty}_{x}\right) \mathrm{d}x.$$
Set
\begin{equation}
\begin{pmatrix}
B\\
C
\end{pmatrix}=\underset{j=1}{\overset{2}{\sum}}\begin{pmatrix}
b_j\\
c_j
\end{pmatrix}\eta^{2-j}   \,,
\end{equation}
and substitute it into \eqref{eq8}. In nonisospectral case ($\eta_t=-2\alpha\eta$ where $\alpha$ is real function of $t$),
letting $A_0=2i \eta^2$, $(b_1,c_1)=-2i(Q,R)$, we obtain the second order nonisospectral AKNS equation
\begin{subequations} \label{nonAKNS}
\begin{align}
&  iQ_t+Q_{xx}-2Q^2R+2i\alpha(xQ)_x=0,\label{AKNSQ}\\
&  iR_t-R_{xx}+2QR^2+2i\alpha(xR)_x=0.\label{AKNSR}
\end{align}
\end{subequations}
With transformation
\begin{equation}
\label{trans 1}
Q=qe^{-\frac{i\alpha}{2}x^2},~~ R=re^{\frac{i\alpha}{2}x^2},
\end{equation}
equation \eqref{nonAKNS} gives rise  to
\begin{subequations} \label{GP-1}
\begin{align}
&  iq_t+q_{xx}-2q^2r+(\delta x^2+i\alpha)q=0,\label{GPq}\\
&  ir_t-r_{xx}+2qr^2-(\delta x^2-i\alpha)r=0,\label{GPr}
\end{align}
\end{subequations}
where $\delta$ is a real constant and here and after the real function $\alpha=\alpha(t)$ is governed by \eqref{a}.
Equation \eqref{GP-1} admits the following reductions
\begin{equation}\label{red}
r(x,t)=\beta q^*(\sigma x,t),~~\beta,\sigma=\pm 1,
\end{equation}
which gives rise to
\begin{equation}
\label{GP00}
iq_t+q_{xx}-2\beta  q^2q^*(\sigma x,t)+(\delta x^2+i\alpha)q=0,
\end{equation}
where $*$ denotes complex conjugate.
When $\beta=-1,\sigma=1$, equation \eqref{GP00} is  the classical  GP equation \eqref{GP2}
while $\sigma=-1$ it is the  nonlocal GP equation
\begin{equation}
\label{nonlocal GP}
iq_t+q_{xx}-2\beta  q^2q^*(-x,t)+(\delta x^2+i\alpha)q=0.
\end{equation}

The transformation \eqref{trans 1} keeps the nonlocal change in \eqref{red},
therefore the following nonlocal nonisospectral NLS from \eqref{nonAKNS} is integrable,
\begin{equation}
\label{nonl-non-NLS}
iQ_t+Q_{xx}-2\beta Q^2Q^*(\sigma x,t)+2i\alpha(xQ)_x=0.
\end{equation}
Classical case of the above equation has been well studied in \cite{Abdselam-2019}.

Note that the classical GP equation \eqref{GP2} admits bilinear form and $N$-soliton solutions.
By transformation $q=g/f$ where $f=f^*$, one can get
the bilinear form of \eqref{GP2}
\begin{subequations}\label{bilinear form 1}
\begin{align}
& (i D_t+D_x^2) g\cdot f=-(\delta x^2+i\alpha)fg, \\
& D_x^2\ f\cdot f=2gg^*,
\end{align}
\end{subequations}
where $D$ is the Hirota bilinear operator defined by \cite{Hirota-1974}
\begin{equation}
 D_x^mD_t^n\ f\cdot g\equiv(\partial_x -\partial_{x^{\prime}})^m (\partial_t-\partial_{t^{\prime}})^n
 f(x,t) g(x^{\prime},t^{\prime})|_{x^{\prime}=x,t^{\prime}=t}.
\end{equation}
Employing the standard procedure of Hirota's method, one can derive 1-,2-,3-soliton solutions for \eqref{bilinear form 1},
which obey the following general form
\begin{subequations}\label{nss}
\begin{align}
& f_N(t,x)=\sum_{\mu=0,1}A_1(\mu)\mathrm{exp}\biggl\{\sum_{j=1}^{2N} \mu_j\xi_j+\sum_{1\leq j<l}^{2N}\mu_j\mu_l\theta_{jl}\biggr\},\\
& g_N(t,x)=\sum_{\mu=0,1}A_2(\mu)\mathrm{exp}\biggl\{\sum_{j=1}^{2N} \mu_j\xi_j+\sum_{1\leq j<l}^{2N}\mu_j\mu_l\theta_{jl}\biggr\},
\end{align}
\end{subequations}
where
\begin{subequations}\label{n}
\begin{align}
& \xi_j=\frac{1}{2}ix^2\sqrt{\delta}\tanh2\sqrt{\delta}t+s_jx\sech2\sqrt{\delta}t
 -\ln\cosh2\sqrt{\delta}t+\frac{is_j^2}{2\sqrt{\delta}}\tanh2\sqrt{\delta}t
 +\xi_j^{(0)},\\
& \xi_{N+j}=\xi^*_j~(j=1,2,\cdots,N),\\
& e^{\theta_{j,N+l}}=\frac{\cosh^2 2\sqrt{\delta}t}{(s_j+s^*_l)^2}~(j,l=1,2,\cdots,N),\\
& e^{\theta_{jl}}=(s_j-s_l)^2\sech2\sqrt{\delta}t,~~e^{\theta_{N+j,N+l}}=e^{\theta^*_{jl}}~(j<l=1,2,\cdots,N),
\end{align}
\end{subequations}
$\xi_j^{(0)}$ and $s_j$ are arbitrary constants,
the summation of $\mu$ means to take all possible $\mu_j=\{0,1\}$ $(j=1,2,\cdots, N)$, and
$A_1(\mu)$ and $A_2(\mu)$ mean that $\mu_j (j=1,2,\cdots, N)$ in the summation of $0$ or 1  meet
\begin{subequations}\label{n}
\begin{align}
& A_1(\mu):\sum_{j=1}^{N}\mu_j=\sum_{j=1}^{N}\mu_{N+j},\\
& A_2(\mu):\sum_{j=1}^{N}\mu_j=1+\sum_{j=1}^{N}\mu_{N+j}.
\end{align}
\end{subequations}
Here and after we only consider the case $\delta > 0$, for the case $\delta < 0$, we can solve it by the same way.

However, the classical Hirota's bilinear operator $D$ does not work for the nonlocal case as $-x$ is involved.
An available  treatment is to consider the unreduced system \eqref{GP-1} rather than \eqref{nonlocal GP}.
Such a reduction technique has been developed in \cite{ChenDLZ-SPAM-2018,ChenZ-AML-2018} recently.

\section{Solutions of the  GP equation}\label{sec-3}

In this section, we apply the reduction technique to construct double Wronskian solutions of both classical and nonlocal GP equations.

\subsection{Solutions of the unreduced system \eqref{GP-1}}

We employ the notation $|\widehat{M-1};\widehat{N-1}|$ introduced in Ref.\cite{Nimmo-NLS-1983} to denote a $(M+N)\times(M+N)$
double Wronskian established as
\begin{equation}
|\widehat{M-1}; \widehat{N-1}|=|\phi^{(M-1)}; \psi^{(N-1)}|=|\phi,\partial_x \phi,\ldots,\partial_x^{M-1} \phi;
\psi,\partial_x \psi,\ldots,\partial_x^{N-1} \psi|,\nonumber
\end{equation}
where $\phi$ and $ \psi $ are $(M+N)$-th column vectors given as
\begin{equation}
\phi=(\phi_1,\phi_2,\ldots,\phi_{M+N})^T,\ \psi=(\psi_1,\psi_2,\ldots,\psi_{M+N})^T. \nonumber
\end{equation}
 When introducing rational transformation
\begin{equation}
q=\frac{g}{f},~~r=\frac{h}{f},
\end{equation}
equation \eqref{GP-1} can be written as bilinear form
\begin{subequations}\label{bilinear form}
\begin{align}
& (i D_t+D_x^2) g\cdot f=-(\delta x^2+i\alpha)gf \label{22a}, \\
& (i D_t-D_x^2) h\cdot f=(\delta x^2-i\alpha)hf \label{22b},\\
& D_x^2\ f\cdot f=-2gh \label{22c}.
\end{align}
\end{subequations}

Employing the Wronskian technique, we have the following theorem.
\begin{theorem}\label{Theorem 1}
The bilinear equations \eqref{bilinear form} allow us the following
double Wronskian solutions
\begin{equation}
\label{wronskian}
f=|\widehat{N-1};\widehat{N-1}|,\quad g=2|\widehat{N-2};\widehat{N}|,
\quad h=-2|\widehat{N};\widehat{N-2}|,
\end{equation}
where the entry vectors $\phi$ and $\psi$ satisfy the  conditions
\begin{subequations}\label{wron-cond-x}
\begin{align}
& \phi_x=A\phi,\quad \phi_t=-2i\phi_{xx} -\frac{i}{2}\delta x^2 \phi,\label{wron-cond-x-a}\\
& \psi_x=-A\psi,\quad  \psi_t=2i\psi_{xx} +\frac{i}{2}\delta x^2 \psi,\label{wron-cond-x-b}
\end{align}
\end{subequations}
in which
\begin{equation}
A=-\frac{i\alpha(t)x}{2}I_{2N}+Q(t),
\label{A}
\end{equation}
where $I_{2N}$ is the $2N\times 2N$  unit matrix, $Q(t)=(Q_{jl}(t))$ is some $2N\times 2N$ matrix depending on $t$
and satisfying
\begin{equation}
Q(t)_t=-2\alpha(t)Q(t),
\label{Q}
\end{equation}
which keeps the compatibility of $\phi_{x,t}=\phi_{t,x}$ and $\psi_{x,t}=\psi_{t,x}$.
\end{theorem}

The proof is long but not trivial, which will be given in Appendix \ref{app-A}.

\subsection{Reductions of double Wronskians}\label{sec-3-2}

As we have seen \eqref{wronskian} provides solutions through double Wronskians $f, g, h$
for the unreduced equation \eqref{GP-1}. Under the reduction $r(x,t)=\beta q^*(\sigma x,t)$, we obtain the classical
and nonlocal GP equations \eqref{GP2} and \eqref{nonlocal GP}. In the following we present a simple reduction
procedure that enables us to obtain double Wronskian solutions both classical and nonlocal GP equations.

\begin{theorem}\label{Theorem 2}
 The classical and nonlocal GP equations \eqref{GP2} and \eqref{nonlocal GP} admit the following solution
\begin{equation}\label{DW}
q=2\frac{|\widehat{N-2};\widehat{N}|}{|\widehat{N-1};\widehat{N-1}|},
\end{equation}
where $\phi$ and $\psi$, as solutions of matrix equations \eqref{wron-cond-x}, are $2N$-th order column vectors,
 and obey the constraint
\begin{equation}\label{T1}
\psi(x)=T\phi^*(\sigma x),
\end{equation}
in which the $2N\times 2N$ constant matrix $T$ determined through
\begin{subequations}\label{T11}
\begin{align}
& A(x)T+\sigma TA^*(\sigma x)=0,\label{T11-A}\\
& TT^*=\beta\sigma I_{2N},\label{T11-T}
\end{align}
\end{subequations}
where $\beta, \sigma=\pm 1$ respectively.
\end{theorem}

\begin{proof}
First, it can be verified that \eqref{wron-cond-x} and \eqref{T1} are compatible under \eqref{T11-A},
 i.e. if we have \eqref{wron-cond-x-a} and \eqref{T11-A}, then $\psi(x)$ defined by \eqref{T1}
 must satisfy the condition \eqref{wron-cond-x-b}.

Next,  we introduce notation (cf.\cite{ChenZ-AML-2018})
\begin{equation}
\widehat{\varphi}^{(s)}(a x)_{[b x]}=\left(\varphi(a x), \partial_{b x} \varphi(a x), \cdots, \partial_{b x}^{s} \varphi(a x)\right),~~ a,b=\pm1,
\end{equation}
and thus, under condition \eqref{T1}, we can rewrite $f, g, h$ in \eqref{wronskian} as
\begin{subequations}
\begin{align}
 &f(x)=|\widehat{\phi}^{(N-1)}(x)_{[x]}; \widehat{\psi}^{(N-1)}(x)_{[x]}|=|\widehat{\phi}^{(N-1)}(x)_{[x]}; T\widehat{\phi}^{*(N-1)}(\sigma x)_{[x]}|,\\
 &g(x)=2|\widehat{\phi}^{(N-2)}(x)_{[x]}; \widehat{\psi}^{(N)}(x)_{[x]}|=2|\widehat{\phi}^{(N-2)}(x)_{[x]}; T\widehat{\phi}^{*(N)}(\sigma x)_{[x]}|,\\
 &h(x)=-2|\widehat{\phi}^{(N)}(x)_{[x]}; \widehat{\psi}^{(N-2)}(x)_{[x]}|=-2|\widehat{\phi}^{(N)}(x)_{[x]}; T\widehat{\phi}^{*(N-2)}(\sigma x)_{[x]}|.
\end{align}
\end{subequations}
By calculation we find
\begin{align*}
f^*(\sigma x)=&|\widehat{\phi}^{(N-1)}(\sigma x)_{[\sigma x]}; T\widehat{\phi}^{*(N-1)}(\sigma^2 x)_{[\sigma x]}|^*\\
=&|\widehat{\phi}^{*(N-1)}(\sigma x)_{[\sigma x]}; T^*\widehat{\phi}^{(N-1)}(x)_{[\sigma x]}|\\
=&(\beta\sigma)^{N}|T^*||T\widehat{\phi}^{*(N-1)}(\sigma x)_{[\sigma x]}; \widehat{\phi}^{(N-1)}(x)_{[\sigma x]}|\\
=&(\beta\sigma)^{N}|T^*|(-1)^{N^2}|\widehat{\phi}^{(N-1)}(x)_{[x]}; T\widehat{\phi}^{*(N-1)}(\sigma x)_{[x]}|\\
=&(\beta\sigma)^{N}(-1)^{N^2}|T^*|f(x),
\end{align*}
and similarly,
\begin{align*}
g^*(\sigma x)
= (\beta\sigma)^{N-1}(-1)^{N^2}\sigma^{2N-1}|T^*|h(x),
\end{align*}
which give rise to
\begin{eqnarray*}
q^*(\sigma x)=\frac{g^*(\sigma x)}{f^*(\sigma x)}
=\frac{(\beta\sigma)^{N-1}(-1)^{N^2}\sigma^{2N-1}|T^*|h(x)}{(\beta\sigma)^{N}(-1)^{N^2}|T^*|f(x)}= \frac{1}{\beta}r(x),
\end{eqnarray*}
i.e. $r(x)=\beta q^*(\sigma x)$. Thus we finish the proof.
\end{proof}

\subsection{Solutions of the classical and nonlocal GP equations}\label{sec-3-3}

In this section, we list Wronskian elements of solutions for the classical and nonlocal GP equations.
We only consider the case $\delta>0$ and hereafter we take $\delta=\nu^2$ where $\nu>0\in \mathbb{R}$.

\subsubsection{Solutions to $Q(t)$ and $T$}\label{sec-3-3-1}

Noticing that due to the form of $A$ in \eqref{A}, the constraint conditions on the $T$ and $A$ in \eqref{T11} can be converted to the conditions on
$T$ and $Q(t)$, i.e.
\begin{subequations}\label{T12}
\begin{align}
& Q(t)T+\sigma TQ^*(t)=0,\label{T12-Q}\\
& TT^*=\beta\sigma I_{2N}.\label{T12-T}
\end{align}
\end{subequations}
As for solutions $T$ and $Q(t)$ of \eqref{T12}, if we assume that they are block matrices of the form
\begin{equation}
T=\left(
                   \begin{array}{cc}
                     T_1 & T_2 \\
                     T_3 & T_4 \\
                   \end{array}
                 \right),~~~
Q(t)=\left(
                   \begin{array}{cc}
                     K_1 & 0 \\
                     0 & K_4 \\
                   \end{array}
                 \right),
\end{equation}
where $T_i$ and $K_i$ are $N\times N$ matrices,
then, solutions to \eqref{T11} are given in Table 1, where $\mathbf{B}_N$ and $\mathbf{H}_N$ are  $N\times N$  matrices:
\begin{table}[h]
\begin{center}
\begin{tabular}{|c|c|c|}
\hline
   ($\beta, \sigma$)   &    $T$  &  $Q(t)$     \\
\hline
     $(1,1)$ &  $T_1=T_4=\mathbf{0}_N, T_3=T_2=\mathbf{I}_N$& $K_1=-K_4^*=\mathbf{B}_N $ \\
\hline
     $(-1,1)$ &  $T_1=T_4=\mathbf{0}_N, T_3=-T_2=\mathbf{I}_N$& $K_1=-K_4^*=\mathbf{B}_N$ \\
                              \hline
     $(1,-1)$ &  $T_1=T_4=\mathbf{0}_N, T_3=-T_2=\mathbf{I}_N$& $K_1=K_4^*=\mathbf{B}_N$ \\
                              \hline
     $(-1,-1)$ &  $T_1=T_4=\mathbf{0}_N, T_3=T_2=\mathbf{I}_N$& $K_1=K_4^*=\mathbf{B}_N$ \\
\hline
\end{tabular}
\caption{$T$ and $Q(t)$ for the GP equation}
\label{Tab-1}
\end{center}
\end{table}

In addition, \eqref{T12} admits more solutions: for the case $(\beta,\sigma)=(1,1)$,
\begin{subequations}\label{real-11}
\begin{align}
& T_1=-T_4=\mathbf{I}_N,~ T_2=T_3=\mathbf{0}_N ~{\rm or}~T=\mathbf{I}_{2N},\\
& K_1=i\mathbf{B}_N,~ K_4=i\mathbf{H}_N,
\end{align}
\end{subequations}
where $\mathbf{B}_N,\mathbf{H}_N\in \mathbb{R}_{N\times N}$, and for the case $(\beta,\sigma)=(-1,-1)$,
\begin{subequations}\label{real-1-1}
\begin{align}
& T_1=-T_4=\mathbf{I}_N,~ T_2=T_3=\mathbf{0}_N ~{\rm or}~T=\mathbf{I}_{2N},\\
& K_1=\mathbf{B}_N,~ K_4=\mathbf{H}_N,
\end{align}
\end{subequations}
where $\mathbf{B}_N,\mathbf{H}_N\in \mathbb{R}_{N\times N}$.
.

\subsubsection{Case by case}\label{sec-3-3-2}

\noindent{Case 1: $\mathbf{B}_N$ being complex diagonal matrix}

When $\mathbf{B}_N$ is diagonal and given by
\begin{equation}\label{diag}
\mathbf{B}_N={\rm Diag}[\gamma_1(t), \gamma_2(t),\cdots, \gamma_N(t)],
\end{equation}
with
\begin{equation}\label{b}
\gamma_j(t)=k_j\sech2 \nu t,~~ k_j\in \mathbb{C}, ~\nu\in \mathbb{R}~(j=1,2,\cdots,N),
\end{equation}
$\phi$ is taken as\footnote{There can be a multiplier $e^{\int \alpha(t) dt}$ in front of $\phi$,
but we have removed it because it contributes nothing to the solution $q=g/f$.}
\begin{equation}\label{phi11}
\phi=\left(a_1^-e^{-\theta_1(x)},a_2^-e^{-\theta_2(x)},\cdots,a_N^-e^{-\theta_N(x)},
a_1^+e^{\theta_1^*(\sigma x)},a_2^+e^{\theta_2^*(\sigma x)},\cdots,a_N^+e^{\theta_N^*(\sigma x)}\right)^T,
\end{equation}
where
\begin{equation}
\theta_j(x)=\frac{i}{4}\alpha x^2-\gamma_j(t)x+2i\int \gamma_j(t)^2dt+ \theta^{(0)}_j
\label{wron-entr-b}
\end{equation}
with $a_j^{\pm}, k_j, \theta^{(0)}_j \in \mathbb{C}$.
Note that $\beta$ takes effects in defining $\psi$.

As examples we list out 1SS for the general GP equation \eqref{GP00}:
\begin{equation}
\label{1-sol-classical-GP21}
|q_1|_{(\beta=\pm1,\sigma=1)}^2=\frac{ 16 a_{1}^2 \sech^2(2\nu t)}{(e^{2X}-\beta e^{-2X})^2},~~~X=c_1-a_1x \sech 2 \nu t -\frac{2a_1b_1}{\nu}\tanh 2\nu t,
\end{equation}
\begin{equation}
\label{1-sol-nonlocal-GP-1}
|q_1|_{(\beta=\pm1,\sigma=-1)}^2=  \frac{8 {b_1}^2 e^{-4 a_1 x \sech 2 \nu t}\,\sech^2 2 \nu t}
{\cosh \left(4 c_1-\frac{8 a_1 b_1}{\nu} \tanh  2 \nu t \right)+\beta \cos \left(4 b_1 x \sech 2 \nu t \right)},
\end{equation}
here and after we take
\[a^{\pm}_j = 1,~ k_j = a_j + ib_j, ~\theta_j^{(0)} = c_j + i d_j\]
for convenience.

\vskip 6pt
\noindent{Case 2: $\mathbf{B}_N$ being Jordan matrix}

When $\mathbf{B}_N$ is a Jordan matrix as follows:
\begin{equation}
\label{B-Jordan}
\mathbf{B}_N=\mathbf{J}_N[\gamma_1(t)]=\begin{pmatrix}
\gamma_1(t)& 0 & 0 & \ldots & 0 & 0\\
\kappa & \gamma_1(t) & 0 & \ldots & 0 & 0\\
\ldots & \ldots & \ldots & \ldots & \ldots & \ldots\\
0 & 0 & 0 & \ldots & \kappa & \gamma_1(t)
\end{pmatrix}_{N\times N},
\end{equation}
where $\gamma_1(t)$ is given as equation \eqref{b}, $\kappa=\sech2 \nu t$, we have
\begin{align}\
\phi=& \biggl(a_1^-e^{-\theta_1(x)},\frac{\partial_{k_1}}{1!}(a_1^-e^{-\theta_1(x)}),\cdots,\frac{\partial_{k_1}^{N-1}}{(N-1)!}(a_1^-e^{-\theta_1(x)}),
\nonumber\\
&~~~~~~ a_1^+e^{\theta_1^*(\sigma x)},\frac{\partial_{k_1^*}}{1!}(a_1^+e^{\theta_1^*(\sigma x)}),\cdots,
\frac{\partial_{k_1^*}^{N-1}}{(N-1)!}(a_1^+e^{\theta_1^*(\sigma x)})\biggr)^T.\label{phi4}
\end{align}

\noindent
\textit{Remark 1~} One can replace the above $\phi$ by $\mathcal{T}\phi$ where $\mathcal{T}$ is diagonal block matrix
$\mathcal{T}=\mathrm{Diag}[\mathcal{A}_N, \mathcal{B}_N]$ with $\mathcal{A}_N$ and $ \mathcal{B}_N$
being $N$-th order arbitrary constant lower triangular Toeplitz matrices (cf.\cite{ZDJ-arxiv,ZhaZSZ-RMP-2014}).

\noindent
\textit{Remark 2~} A general case for $\mathbf{B}_N$ is the diagonal block form
\begin{equation}
\mathbf{B}_N=\mathrm{Diag}\left(\mathbf{J}_{h_1}[\gamma_1(t)],\mathbf{J}_{h_2}[\gamma_2(t)],\cdots,\mathbf{J}_{h_s}[\gamma_s(t)],
\mathrm{Diag}[\gamma_{s+1}(t),\cdots, \gamma_{s+m}(t)]\right),
\end{equation}
where each $\mathbf{J}_{h_j}[\gamma_j(t)]$ is an $h_j\times h_j$ Jordon block matrix defined as \eqref{B-Jordan},
$\mathrm{Diag}[\gamma_{s+1}(t),\cdots, \gamma_{s+m}(t)]$ is an $m\times m$ diagonal matrix and $\sum_{j=1}^{s} h_j+m=N$.
In this case, $\phi$ is just composed accordingly since \eqref{wron-cond-x} is a linear system of $\phi$.

\vskip 6pt
\noindent{Case 3: $\mathbf{B}_N,\mathbf{H}_N$ being real matrices }

For the case \eqref{real-11} where $(\beta,\sigma)=(1,1)$, the diagonal $\mathbf{B}_N$ and $\mathbf{H}_N$ are
\begin{equation}\label{B,H}
\mathbf{B}_N={\rm Diag}[\gamma_1(t),\gamma_2(t),\cdots, \gamma_N(t)],
~~\mathbf{H}_N={\rm Diag}[\omega_1(t),\omega_2(t),\cdots, \omega_N(t)]
\end{equation}
with
\begin{equation}\label{b,h}
\gamma_j(t)=k_j\sech2\nu t,~~ \omega_j(t)=l_j\sech2\nu t,~~ k_j,\l_j\in \mathbb{R}, ~(j=1,2,\cdots,N),
\end{equation}
and we have
\begin{equation}
\phi=\left(a_1^-e^{-\theta_1},a_2^-e^{-\theta_2},\cdots,a_N^-e^{-\theta_N},
a_1^+e^{-\rho_1},a_2^+e^{-\rho_2},\cdots,a_N^+e^{-\rho_N}\right)^T,
\end{equation}
where
\begin{equation}\label{theta,rho}
\theta_j=\frac{i}{4}\alpha x^2-i \gamma_j(t)x- 2i\int \gamma_j(t)^2dt+ \theta^{(0)}_j,
~~\rho_j=\frac{i}{4}\alpha x^2-i \omega_j(t)x-2i\int \omega_j(t)^2dt+ \rho^{(0)}_j
\end{equation}
with $\theta^{(0)}_j, \rho^{(0)}_j\in \mathbb{C}$.
In the case $T_1=-T_4=\mathbf{I}_N,~ T_2=T_3=\mathbf{0}_N$ and $N=1$ we have
\begin{equation}
|q_1|_{(\beta=1,\sigma=1)}^2=\frac{(\gamma_1(t)-\omega_1(t))^2}{\cos^2(V_1-V_2)},
\end{equation}
where
\begin{equation}\label{A1,B1}
V_1=\frac{1}{4}\alpha x^2-\gamma_1(t)x- 2\int \gamma_1(t)^2dt+d_1,~~~V_2=\frac{1}{4}\alpha x^2-\omega_1(t)x- 2\int \omega_1(t)^2dt+n_1
\end{equation}
and we have taken $a^{\pm}_j = 1$, $\theta_j^{(0)} = c_j + i d_j$ and $\rho_j^{(0)} = m_j + i n_j$ for convenience.
In the case $T=\mathbf{I}_{2N}$, we get
\begin{equation}
|q_1|_{(\beta=1,\sigma=1)}^2=\frac{(\gamma_1(t)-\omega_1(t))^2}{\sin^2(V_1-V_2)},
\end{equation}
with $V_1,V_2$ defined in \eqref{A1,B1}.
If $\mathbf{B}_N$ and $\mathbf{H}_N$ are corresponding Jordan form as in \eqref{B-Jordan},
we have
\begin{align}
\phi=&\biggl(a_1^-e^{-\theta_1},\frac{\partial_{k_1}}{1!}(a_1^-e^{-\theta_1}),\cdots,\frac{\partial_{k_1}^{N-1}}{(N-1)!}(a_1^-e^{-\theta_1}),\nonumber\\
&~~~~ a_1^+e^{-\rho_1},\frac{\partial_{l_1}}{1!}(a_1^+e^{-\rho_1}),\cdots,\frac{\partial_{l_1}^{N-1}}{(N-1)!}(a_1^+e^{-\rho_1})\biggr)^T \label{51}
\end{align}
with $\theta_1, \rho_1$ defined in \eqref{theta,rho}.

For the case \eqref{real-1-1} where $(\beta,\sigma)=(-1,-1)$,
$\mathbf{B}_N$ and $\mathbf{H}_N$ are diagonal matrices \eqref{B,H} with elements \eqref{b,h},
we have
\begin{equation}
\phi=\left(a_1^-e^{-\theta_1},a_2^-e^{-\theta_2},\cdots,a_N^-e^{-\theta_N},
a_1^+e^{-\rho_1},a_2^+e^{-\rho_2},\cdots,a_N^+e^{-\rho_N}\right)^T,
\end{equation}
where
\begin{equation}
\theta_j=\frac{i}{4}\alpha x^2-\gamma_j(t)x+2i\int \gamma_j(t)^2dt+ \theta^{(0)}_j,
~~\rho_j=\frac{i}{4}\alpha x^2-\omega_j(t)x+2i\int \omega_j(t)^2dt+ \rho^{(0)}_j
\label{53}
\end{equation}
with $\gamma_j(t), \omega_j(t)$ defined in \eqref{b,h} and $\theta^{(0)}_j,\rho^{(0)}_j\in \mathbb{C}$.
When $T_1=-T_4=\mathbf{I}_N,~ T_2=T_3=\mathbf{0}_N$   we can get solution
\begin{equation}
|q_1|_{(\beta=-1,\sigma=-1)}^2=\frac{4(\gamma_1(t)-\omega_1(t))^2e^{-2(\gamma_1(t)+\omega_1(t))x}}{e^{2(\gamma_1(t)-\omega_1(t))x}
+e^{-2(\gamma_1(t)-\omega_1(t))x}+2\cos 2(W_1-W_2)},
\end{equation}
and if  $T=\mathbf{I}_{2N}$ we have
\begin{equation}
|q_1|_{(\beta=-1,\sigma=-1)}^2=\frac{4(\gamma_1(t)-\omega_1(t))^2e^{-2(\gamma_1(t)+\omega_1(t))x}}{e^{2(\gamma_1(t)-\omega_1(t))x}
+e^{-2(\gamma_1(t)-\omega_1(t))x}-2\cos2(W_1-W_2)},
\end{equation}
where
\begin{equation}
W_1=\frac{1}{4}\alpha x^2+ 2\int \gamma_1(t)^2dt+d_1,~~~W_2=\frac{1}{4}\alpha x^2+2\int \omega_1(t)^2dt+n_1.
\end{equation}
If $\mathbf{B}_N$ and $\mathbf{H}_N$ are corresponding Jordan form as in \eqref{B-Jordan},
$\phi$ takes the form \eqref{51} but where $\theta_1$ and $\rho_1$ are given by \eqref{53}.

\section{Localized dynamics of the classical and nonlocal GP equations}\label{sec-4}

\subsection{Classical case}\label{sec-4-1}

\subsubsection{1SS}\label{sec-4-1-1}

It is interesting that many solutions we derived show localized characteristics in both space and time.
Consider the classical GP equation \eqref{GP2} with $\delta>0$.
We can rewrite 1SS \eqref{1-sol-classical-GP21} with $\beta=-1$, which reads
\begin{equation}\label{1-sol-GP2}
|q_1|_{(\beta=-1,\sigma=1)}^2= 4a_{1}^2 \sech ^2(2\nu t)~
\sech^2 \left [2a_{1}\sech  2(\nu t) \left(x+\frac{2b_1}{\nu}\sinh 2\nu t-h_1 \cosh 2\nu t\right) \right],
\end{equation}
where $h_1=c_1/a_1$.
This solution provides a bell-shaped soliton traveling with a localized time dependent amplitude $4a_{1}^2 \sech ^2(2\nu t)$, top trajectory
\begin{equation}
x(t)=-\frac{2b_1}{\nu}\sinh 2\nu t+h_1\cosh2\nu t
=\left(-\frac{ b_1}{\nu}+\frac{h_1 }{2 }\right)e^{2\nu t}+ \left(\frac{ b_1}{\nu}+\frac{h_1 }{2 }\right) e^{-2 \nu t},
\label{top trace-GP}
\end{equation}
velocity
\begin{equation}
x^\prime (t)=(\nu h_1-2b_1)e^{2\nu t}- (\nu h_1+2b_1)e^{-2\nu t},
\label{velocity-GP}
\end{equation}
and vertex at $(x,t)=(h_1,0)$.
The above top trajectory can be further described according to the sign of $s=(\nu h_1-2b_1)(\nu h_1+2b_1)$:
when $s>0$, $x(t)$ runs like $\cosh 2 \nu t $, when $s<0$, $x(t)$ runs like $\sinh 2 \nu t $,
and when $s=0$, $x(t)$ runs like $e^{2 \nu t}$ or $e^{-2\nu t}$ or in particular $x(t)$ is stationary when $b_1=c_1=0$.
Like the solitons of the NLS equation,
the amplitude and velocity for \eqref{1-sol-GP2} are governed by two independent parameters,
i.e. the real part $a_1$ and imaginary part $b_1$ of $k_1$,
which differs from the KdV solitons.
Fig.\ref{fig-1}(a) and Fig.\ref{fig-1}(b) depict a moving localized wave and a stationary one, respectively.


\captionsetup[figure]{labelfont={bf},name={Fig.},labelsep=period}
\begin{figure}[ht]
\centering
\subfigure[ ]{
\begin{minipage}[t]{0.35\linewidth}
\centering
\includegraphics[width=2.1in]{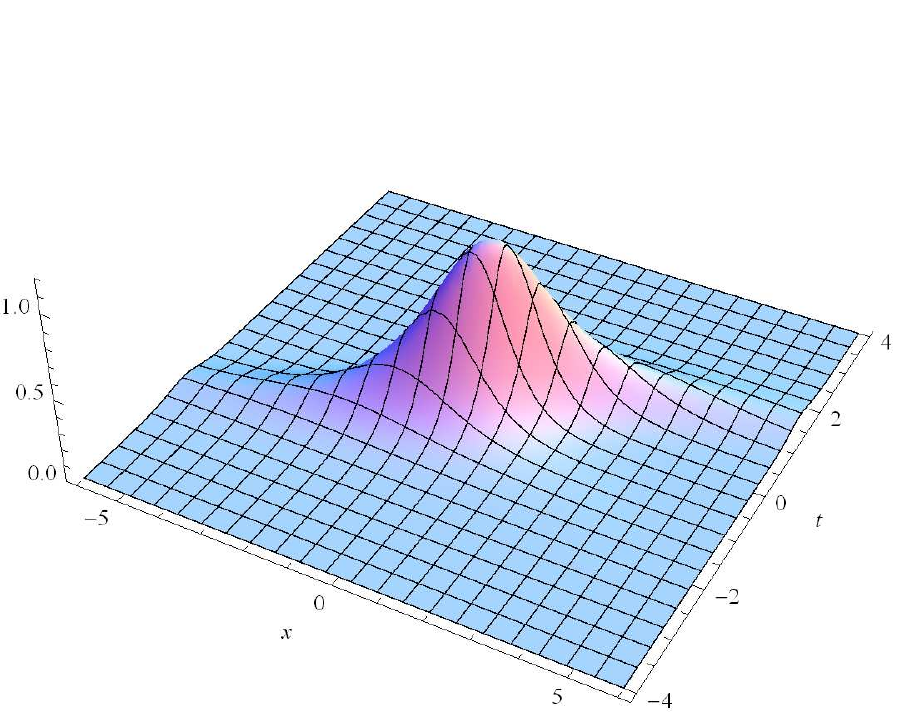}
\end{minipage}%
}%
\subfigure[ ]{
\begin{minipage}[t]{0.35\linewidth}
\centering
\includegraphics[width=2.1in]{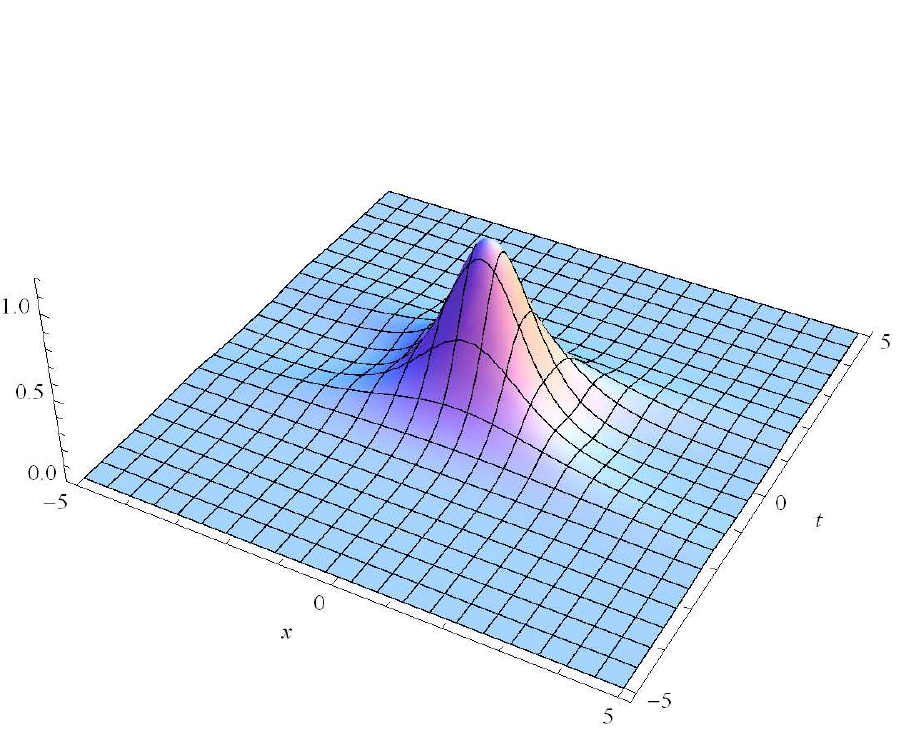}
\end{minipage}%
}
\caption{Shape and motion of 1SS to the classical GP equation \eqref{GP2}.
(a) A moving soliton given by \eqref{1-sol-GP2} for $k_1=0.5-0.3i$, $\theta^{(0)}_1=0$ and $\nu=0.6$.~
(b) A stationary soliton given by \eqref{1-sol-GP2} for $k_1=0.5$, $\theta^{(0)}_1=0$ and $\nu=0.6$.~
}
\label{fig-1}
\end{figure}


\subsubsection{2SS}\label{sec-4-1-2}

Fig.\ref{fig-2}  shows a head-on collision of two solitons given by
\begin{equation}\label{2-sol-classical-GP-a}
|q_2|_{(\beta=-1,\sigma=1)}^2=\frac{g g^*}{f^2},~~~ f=|\phi,~\partial_x \phi; ~\psi,~ \partial_x \psi|,~~~
g=2|\phi;~ \psi, ~\partial_x \psi, ~\partial^2_x \psi|
\end{equation}
with
\begin{equation}\label{2-sol-classical-GP-b}
\phi=\bigl(e^{-\theta_1},e^{-\theta_2};e^{\theta_1^*},e^{\theta_2^*}\bigr)^T,~~~
\psi=\bigl(-e^{\theta_1},-e^{\theta_2};e^{-\theta_1^*},e^{-\theta_2^*}\bigr)^T.
\end{equation}


\captionsetup[figure]{labelfont={bf},name={Fig.},labelsep=period}
\begin{figure}[ht]
\centering
\subfigure[ ]{
\begin{minipage}[t]{0.35\linewidth}
\centering
\includegraphics[width=2.1in]{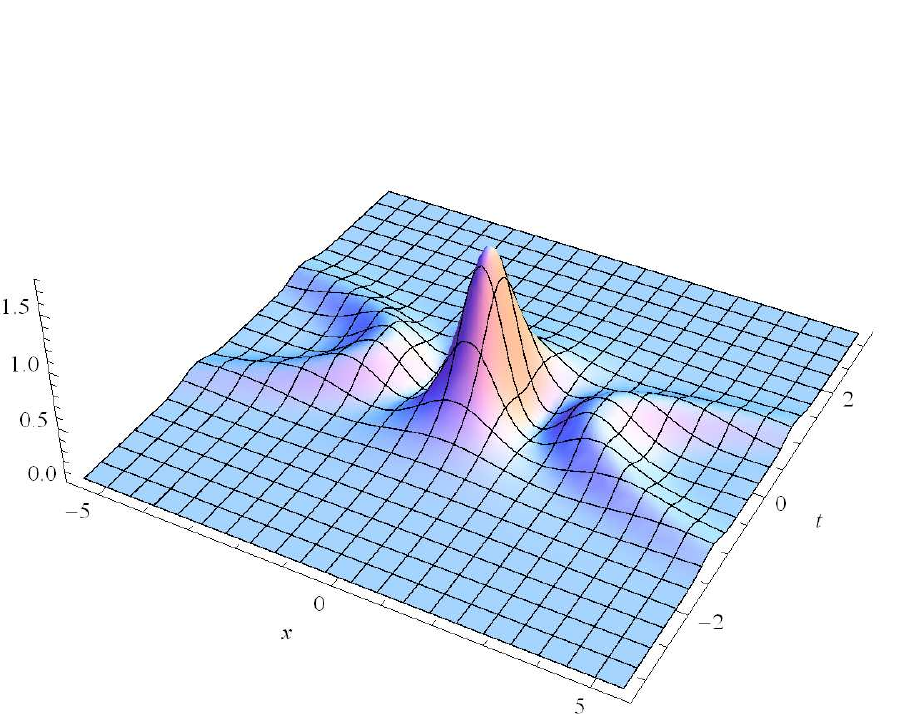}
\end{minipage}%
}%
\subfigure[ ]{
\begin{minipage}[t]{0.35\linewidth}
\centering
\includegraphics[width=2.1in]{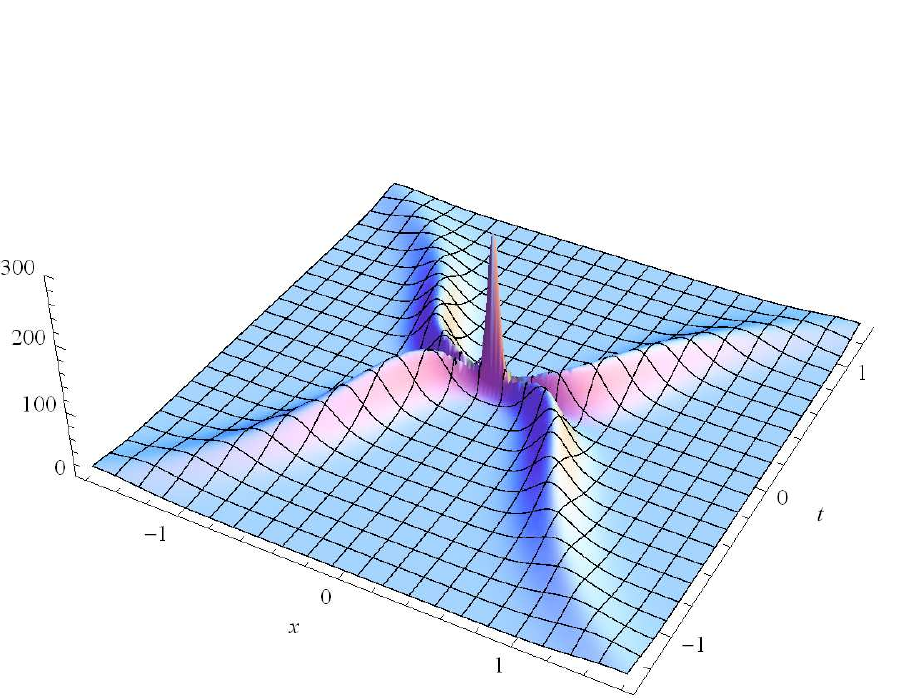}
\end{minipage}%
}
\caption{Two-solition interactions of the classical GP equation \eqref{GP2}.
(a) A head-on collision of two solitons
given by \eqref{2-sol-classical-GP-a} and \eqref{2-sol-classical-GP-b} for $k_1=0.3+0.5i$, $k_2=0.3-0.5i$, $\theta^{(0)}_1=\theta^{(0)}_2=0$ and $\nu=0.6$.~
(b) Non-periodic interaction of two solitions with same speed,  given by \eqref{2-sol-classical-GP-a} and \eqref{2-sol-classical-GP-b}
for $k_1=4, k_2=4.1$, $\theta^{(0)}_1=\theta^{(0)}_2=0$ and $\nu=0.6$.
}
\label{fig-2}
\end{figure}


When $b_1=b_2=b$ as well as $h_1 =h_2 =h$ but $a_1\neq a_2$ in \eqref{2-sol-classical-GP-b}, we have
\begin{equation}
|q_2(x,t)|_{(\beta=-1,\sigma=1)}^2
=\frac{16(a^2_1-a^2_2)^2 u(x,t) \sech ^2(2\nu t)}{v(x,t)^2},
\end{equation}
with
\begin{align*}
u(x,t)= &~2a^2_1+2a^2_2+a^2_1\left(e^{2Y_2}+e^{-2Y_2}\right)+a^2_2\left(e^{2Y_1}+e^{-2Y_1}\right) \nonumber\\
& -2a_1 a_2(e^{Y_1+Y_2}+e^{-(Y_1+Y_2)}+e^{Y_1-Y_2}+e^{Y_2-Y_1})\cos2(Z_1-Z_2),
\end{align*}
\begin{align*}
v(x,t)=& e^{-\frac{1}{2}(X_1+X_4)}\left[(a^2_1+a^2_2)\left(e^{X_1}+e^{X_2}+e^{X_3}+e^{X_4}\right)-2a_1 a_2\left(e^{X_1}-e^{X_2}-e^{X_3}+e^{X_4}\right)\right]\\
&~~ -8a_1 a_2 \cos2X_5,
\end{align*}
where
\begin{align*}
&Y_i=-2a_{i}h+2a_{i}x \sech 2\nu t+\frac{4a_{i}b\tanh 2\nu t}{\nu},\\
& Z_i=d_i-bx\sech2\nu t+\frac{(4a_{i}^2-b^2+4\nu^2 x^2)\tanh2\nu t}{\nu}~(i=1,2),\\
& X_1=4(a_1+a_2)x\sech2\nu t+\frac{4(a_1+a_2)b\tanh2\nu t}{\nu},\\
& X_2=4a_1h+4a_2x\sech2\nu t-\frac{4(a_1-a_2)b\tanh2\nu t}{\nu},\\
& X_3=4a_2h+4a_1x\sech2\nu t-\frac{4(a_2-a_1)b\tanh2\nu t}{\nu},\\
& X_4=4(a_1+a_2)h-\frac{4(a_1+a_2)b\tanh2\nu t}{\nu},\\
& X_5=d_1-d_2+(a^2_1-a^2_2)\frac{\tanh 2\nu t}{\nu}.
\end{align*}
In this case, the two solitons can travel with  same velocity
$$x^\prime (t)=(\nu h-2b) e^{2\nu t}-(\nu h+2b) e^{-2\nu t},$$
and same top trajectory
\begin{equation}
x(t)=-\frac{2b}{\nu}\sinh 2\nu t+h\cosh2\nu t.
\label{top trace-NNLS-2(2-sol)}
\end{equation}
In addition, the value of $|q_2|^2$ on the curve \eqref{top trace-NNLS-2(2-sol)} is given as
\begin{equation}\label{2-sol-GP2-period}
|q_2(-\frac{2b}{\nu}\sinh 2\nu t+h\cosh2\nu t,t)|^2
=\frac{4(a^2_1-a^2_2)^2\sech^2(2\nu t)}{a_1^2 +a^2_2
-2a_1 a_2\cos \left[2(d_1-d_2)+2(a^2_1-a^2_2)\frac{\tanh 2\nu t}{\nu}\right]},
\end{equation}
which allows variety of interaction behaviors.

First, noting that $|\tanh 2\nu t|<1$, when $|\varpi|\geq 1$ where
\begin{equation}
\varpi=\frac{\nu(\pi-d_1+d_2)}{a_1^2-a_2^2},
\end{equation}
there is no periodic interaction, as depicted in Fig.\ref{fig-2}(b).
Second, periodic-like interaction can happen when $|\varpi|<1$. In this case, the wave oscillates in each interval
$[t_s,t_{s+1}]$ where
$t_s=\frac{1}{2\nu}\mathrm{artanh}\Bigl[\frac{\nu(s  \pi - (d_1-d_2))}{a_1^2-a^2_2}\Bigr]$ if $a_1\cdot a_2 >0$
and $t_s=\frac{1}{2\nu}\mathrm{artanh}\Bigl[\frac{\nu((s+\frac{1}{2})  \pi - (d_1-d_2))}{a_1^2-a^2_2}\Bigr]$ if $a_1\cdot a_2 <0$.
The dimension of each oscillation is estimated $16a_1a_2\sech^2(2\nu t_s)$.
It is interesting that \eqref{2-sol-GP2-period} indicates the oscillation is carried by a bell-shape wave.
We describe two stationary ($x(t)=0$) periodic-like solutions in Fig.\ref{fig-3} with large oscillation and small oscillation.


\captionsetup[figure]{labelfont={bf},name={Fig.},labelsep=period}
\begin{figure}[ht]
\centering
\subfigure[ ]{
\begin{minipage}[t]{0.4\linewidth}
\centering
\includegraphics[width=2.1in]{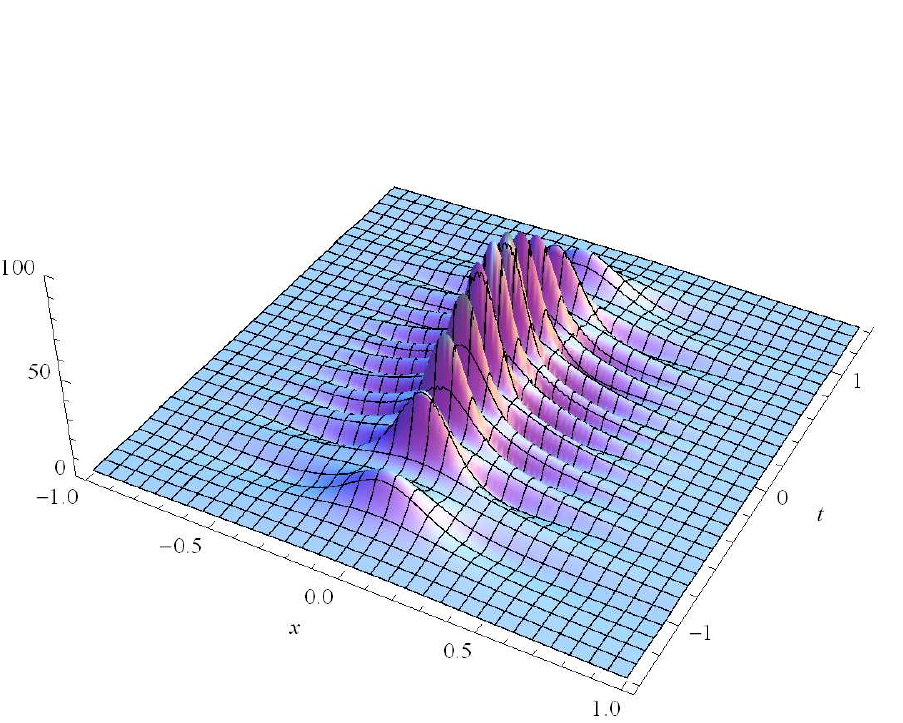}
\end{minipage}%
}%
\subfigure[ ]{
\begin{minipage}[t]{0.4\linewidth}
\centering
\includegraphics[width=2.1in]{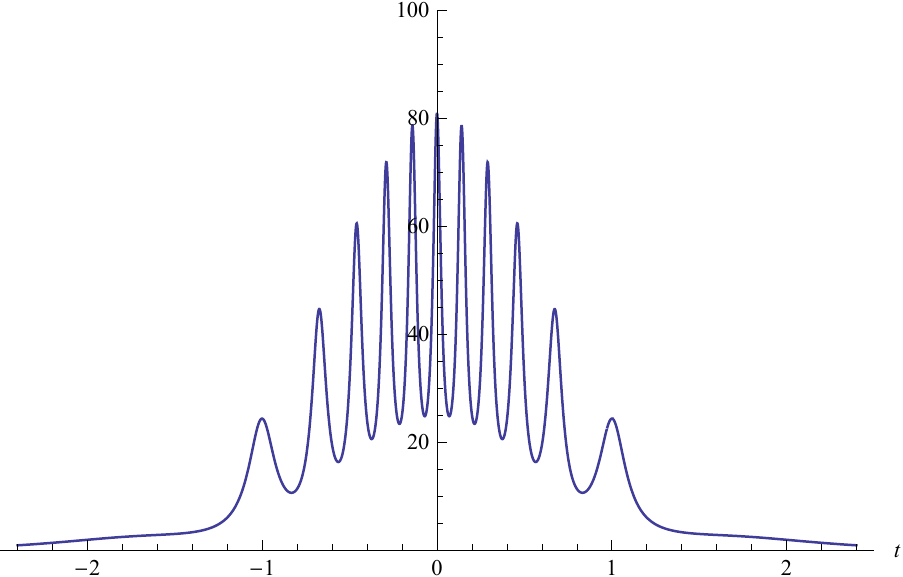}
\end{minipage}%
}%
\\
\subfigure[ ]{
\begin{minipage}[t]{0.4\linewidth}
\centering
\includegraphics[width=2.1in]{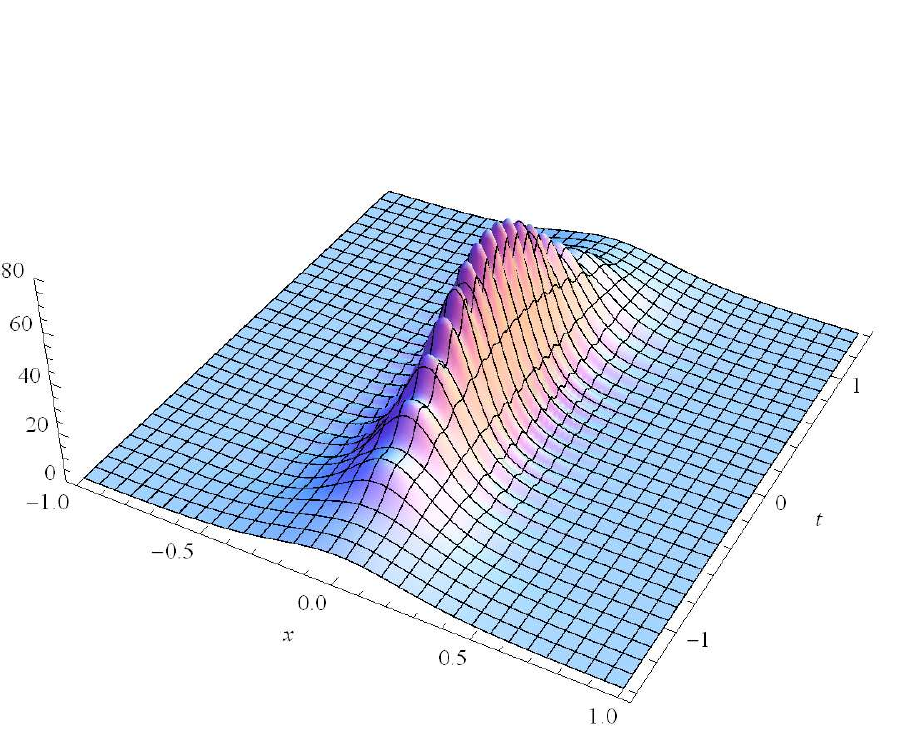}
\end{minipage}%
}%
\subfigure[ ]{
\begin{minipage}[t]{0.4\linewidth}
\centering
\includegraphics[width=2.1in]{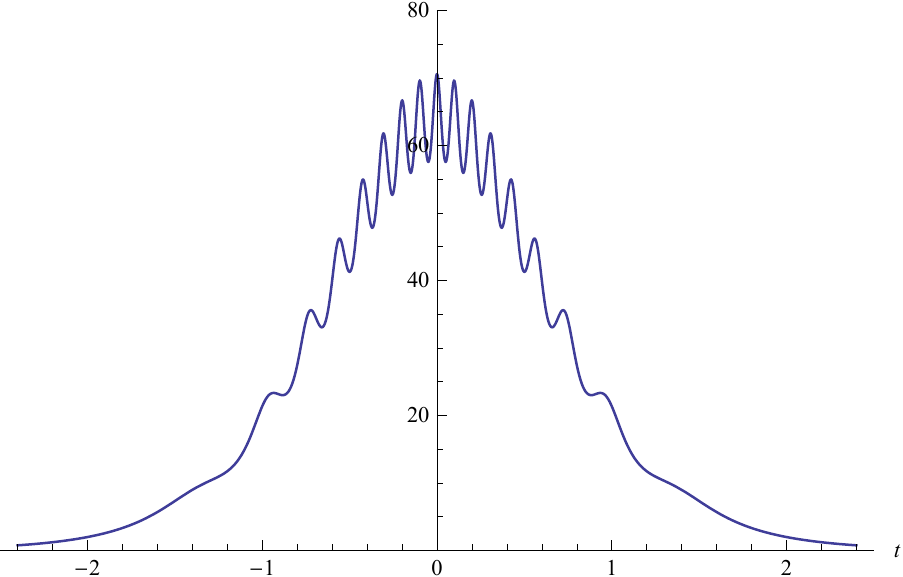}
\end{minipage}%
}%
\caption{Periodic-like interactions of the classical GP equation \eqref{GP2}.
(a) Soliton given by equation \eqref{2-sol-GP2-period} for $k_1=3.5$, $k_2=1$, $\theta^{(0)}_1=\theta^{(0)}_2=0$ and $\nu=0.6$.
(b) The 2D plot of (a) at $x=0$.
(c) Soliton given by equation \eqref{2-sol-GP2-period} for $k_1=4$, $k_2=0.2$, $\theta^{(0)}_1=\theta^{(0)}_2=0$ and $\nu=0.6$.
(d) The 2D plot of (c) at $x=0$.
}
\label{fig-3}
\end{figure}


\subsubsection{Jordan block solution and 3SS}\label{sec-4-1-3}

The simplest Jordan block solution of the classical GP equation \eqref{GP2}
is given by \eqref{phi4} with
\begin{equation}\label{2-sol-Jordan}
\phi=\bigl(e^{-\theta_1},\partial_{k_1}e^{-\theta_1},
e^{\theta_1^*},\partial_{k_1^*}e^{\theta_1^*}\bigr)^T,~~
\psi=\bigl(-e^{\theta_1},-(\partial_{k_1^*}e^{\theta_1^*})^*,
e^{-\theta_1^*},(\partial_{k_1}e^{-\theta_1})^*\bigr)^T,
\end{equation}
and $\theta_1$ is defined as \eqref{wron-entr-b},
the solution is described in Fig.~\ref{fig-4}.


\captionsetup[figure]{labelfont={bf},name={Fig.},labelsep=period}
\begin{figure}[h]
\centering
\subfigure[ ]{
\begin{minipage}[t]{0.45\linewidth}
\centering
\includegraphics[width=2.1in]{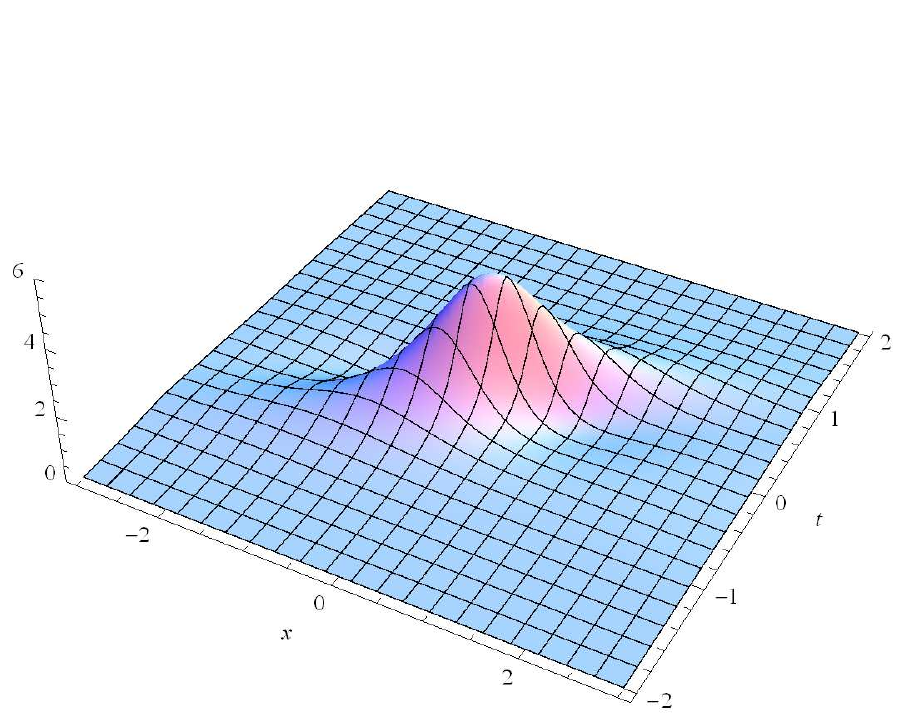}
\end{minipage}%
}%
\subfigure[ ]{
\begin{minipage}[t]{0.45\linewidth}
\centering
\includegraphics[width=2.1in]{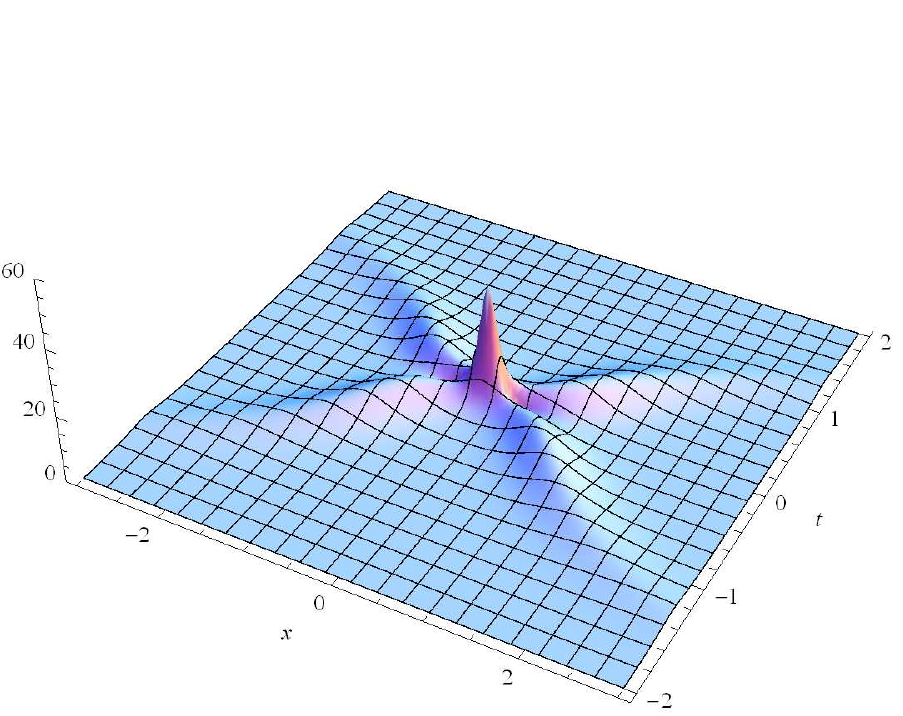}
\end{minipage}%
}%
\caption{Shape and motion for Jordan block solution to the classical GP equation \eqref{GP2}.
(a) Jordan block solution given by \eqref{2-sol-Jordan}
with $k_1=0.5-0.3i$, $\theta^{(0)}_1=0$ and $\nu=0.6$. (b) Jordan block solution given by \eqref{2-sol-Jordan}
with $k_1=1.5$, $\theta^{(0)}_1=0$ and $\nu=0.6$.}
\label{fig-4}
\end{figure}


We can also consider three solitons for classical GP equation \eqref{GP2} given by
\begin{equation}\label{3-sol-classical-GP-a}
|q_3|_{(\beta=-1,\sigma=1)}^2=\frac{g g^*}{f^2},~~ f=\left|\phi,~\partial_x \phi,~\partial^2_x \phi; ~\psi,~ \partial_x \psi,~\partial^2_x \psi\right|,~~
g=2\left|\phi,~\partial_x \phi; ~ \psi, ~\partial_x \psi, ~\partial^2_x \psi,~\partial^3_x \psi\right|
\end{equation}
with
\begin{equation}\label{3-sol-classical-GP-b}
\phi=\bigl(e^{-\theta_1},e^{-\theta_2},e^{-\theta_3};e^{\theta_1^*},e^{\theta_2^*},e^{\theta_3^*}\bigr)^T,~~~
\psi=\bigl(-e^{\theta_1},-e^{\theta_2},-e^{\theta_3};e^{-\theta_1^*},e^{-\theta_2^*},e^{-\theta_3^*}\bigr)^T.
\end{equation}
Fig.~\ref{fig-5}(a) and (b) show the interactions between one moving soliton and two solitons with same velocity,
and one can see a clear phase shift resulted from collision.
Fig.~\ref{fig-5}(c) and (d) depict a stationary soliton
comes into collision with the periodical ones and without phase shift for the later.


\captionsetup[figure]{labelfont={bf},name={Fig.},labelsep=period}
\begin{figure}[h]
\centering
\subfigure[ ]{
\begin{minipage}[t]{0.4\linewidth}
\centering
\includegraphics[width=2.1in]{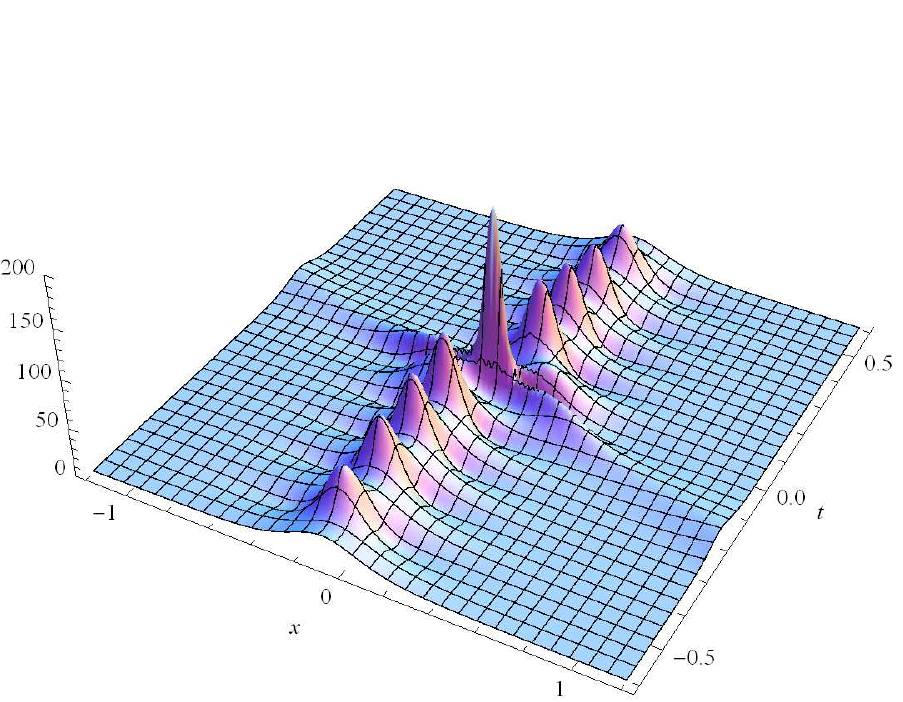}
\end{minipage}%
}%
\subfigure[ ]{
\begin{minipage}[t]{0.4\linewidth}
\centering
\includegraphics[width=2.1in]{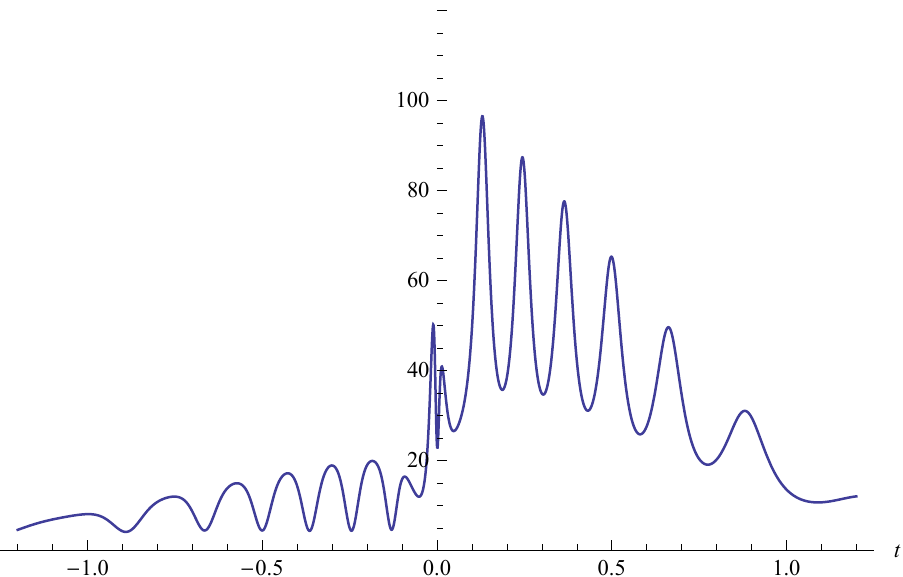}
\end{minipage}%
}%
\\
\subfigure[ ]{
\begin{minipage}[t]{0.4\linewidth}
\centering
\includegraphics[width=2.1in]{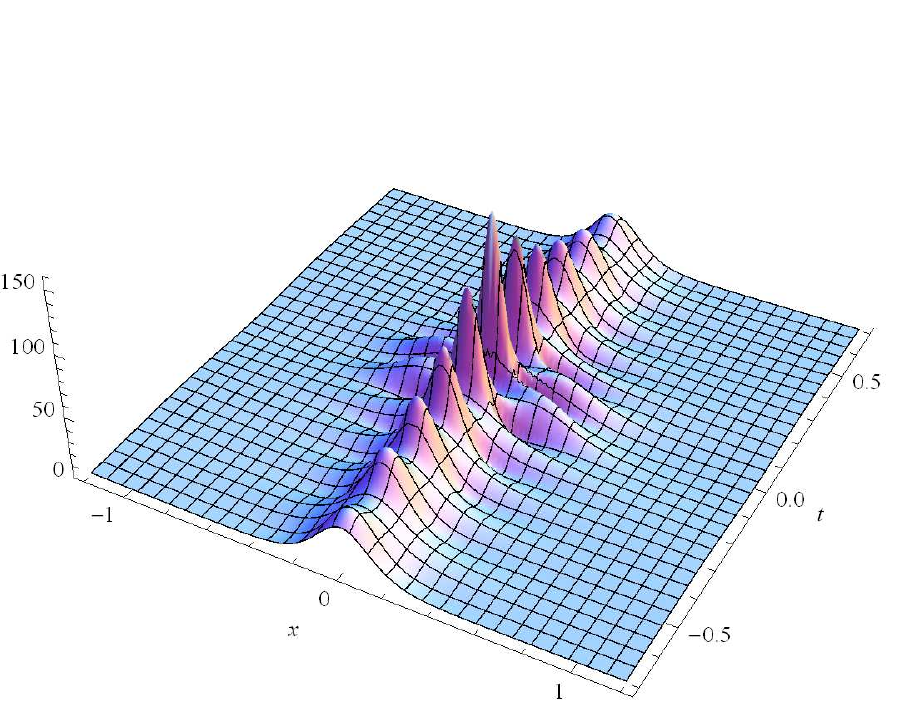}
\end{minipage}%
}%
\subfigure[ ]{
\begin{minipage}[t]{0.4\linewidth}
\centering
\includegraphics[width=2.1in]{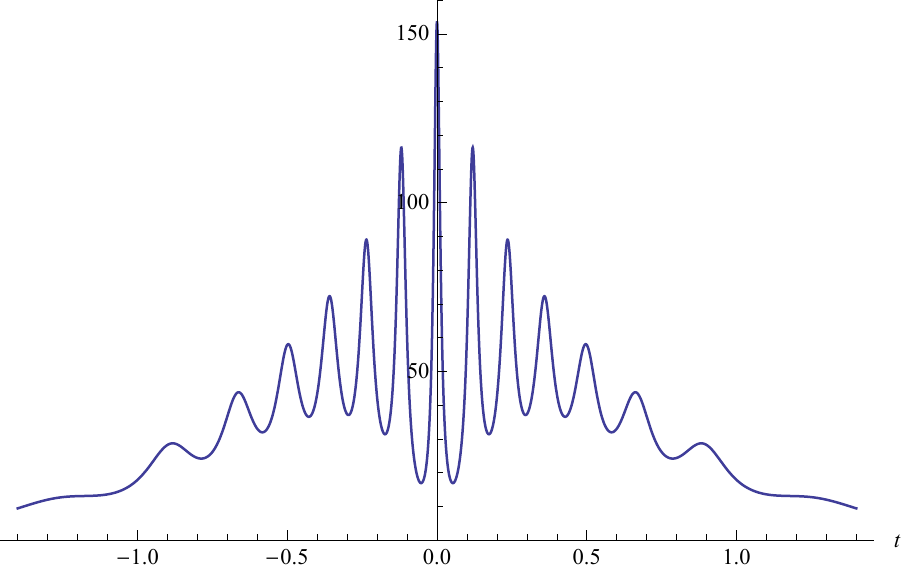}
\end{minipage}%
}%
\caption{Shape and motion of 3SS to the classical GP equation \eqref{GP2}.
(a) Interactions of three solitons given by equation \eqref{3-sol-classical-GP-a}
and \eqref{3-sol-classical-GP-b} for $k_1=4$, $k_2=1$, $k_3=2+i$, $\theta^{(0)}_1=\theta^{(0)}_2=\theta^{(0)}_3=0$ and $\nu=0.6$.
(b) The plot of $|q_3|^2$ on $x=0.1$ for (a).
(c) Interactions of three solitons given by equation \eqref{3-sol-classical-GP-a}
and \eqref{3-sol-classical-GP-b} for $k_1=4$, $k_2=1$, $k_3=1.2$, $\theta^{(0)}_1=\theta^{(0)}_2=\theta^{(0)}_3=0$ and $\nu=0.6$.
(d) The plot of $|q_3|^2$ on $x=0$ for (c).
}
\label{fig-5}
\end{figure}



\subsection{Nonlocal case}\label{sec-4-2}

\subsubsection{1SS}\label{sec-4-2-1}

Next, let us look at soliton solutions for the case $(\beta,\sigma)=(1,-1)$ of the nonlocal GP equation \eqref{GP00}.
One-soliton solution is given by \eqref{1-sol-nonlocal-GP-1} with $\beta=1$, which we depict in
Fig.~\ref{fig-6}. Noting that when $b_1=0$, we get trivial solution 0.


\captionsetup[figure]{labelfont={bf},name={Fig.},labelsep=period}
\begin{figure}
\centering
\includegraphics[height=4cm,width=6cm]{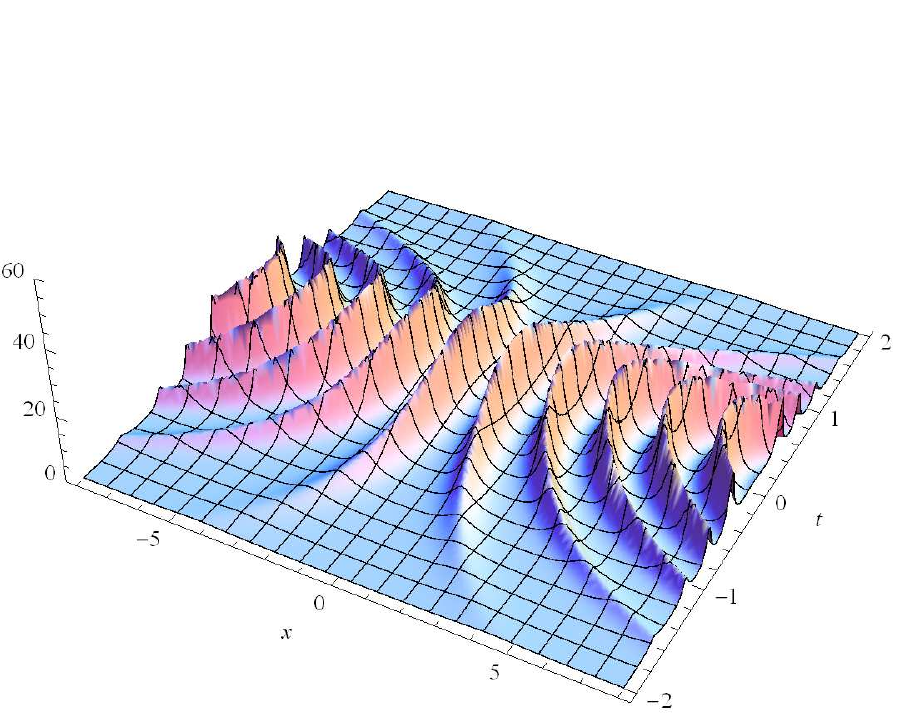}
\caption{ Shape and motion of 1SS \eqref{1-sol-nonlocal-GP-1} for equation \eqref{nonlocal GP} with $\beta=1$, in which
$k_1=i$, $\theta^{(0)}_1=0.2$ and $\nu=0.6$.}
\label{fig-6}
\end{figure}


\subsubsection{2SS}\label{sec-4-2-2}

Two-soliton solutions are obtained when we take
\begin{equation}\label{non-2SS-b}
\phi=\bigl(e^{-\theta_1},e^{-\theta_2};e^{\theta_1^*(-x)},e^{\theta_2^*(-x)}\bigr)^T,~~~
\psi=\bigl(-e^{\theta_1},-e^{\theta_2};e^{-\theta_1^*(-x)},e^{-\theta_2^*(-x)}\bigr)^T,
\end{equation}
we show by Fig.~\ref{fig-7}(a) and Fig.~\ref{fig-7}(b),
from which we can see the quasi-periodic interaction in 1SS case still exists in
Fig.~\ref{fig-7}(a), but it vanishes when $x$ is large enough, Fig.~\ref{fig-7}(b)
shows that when we choose a slightly larger $a_1$, the quasi-periodic interaction does not occur,
which is similar to the classical case.


\captionsetup[figure]{labelfont={bf},name={Fig.},labelsep=period}
\begin{figure}[ht]
\centering
\subfigure[ ]{
\begin{minipage}[t]{0.35\linewidth}
\centering
\includegraphics[width=2.1in]{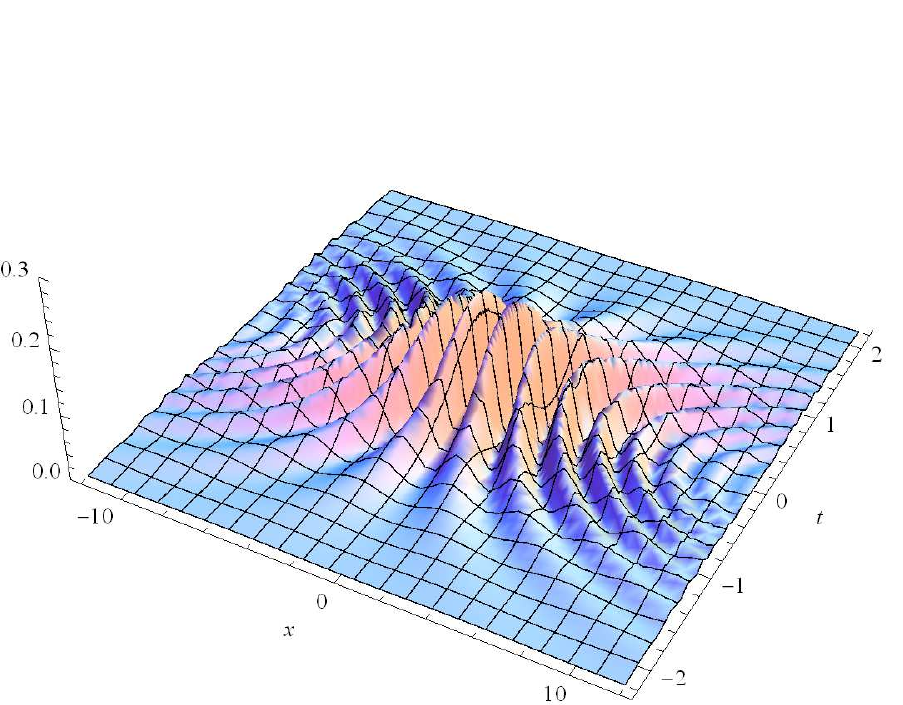}
\end{minipage}%
}%
\subfigure[ ]{
\begin{minipage}[t]{0.35\linewidth}
\centering
\includegraphics[width=2.1in]{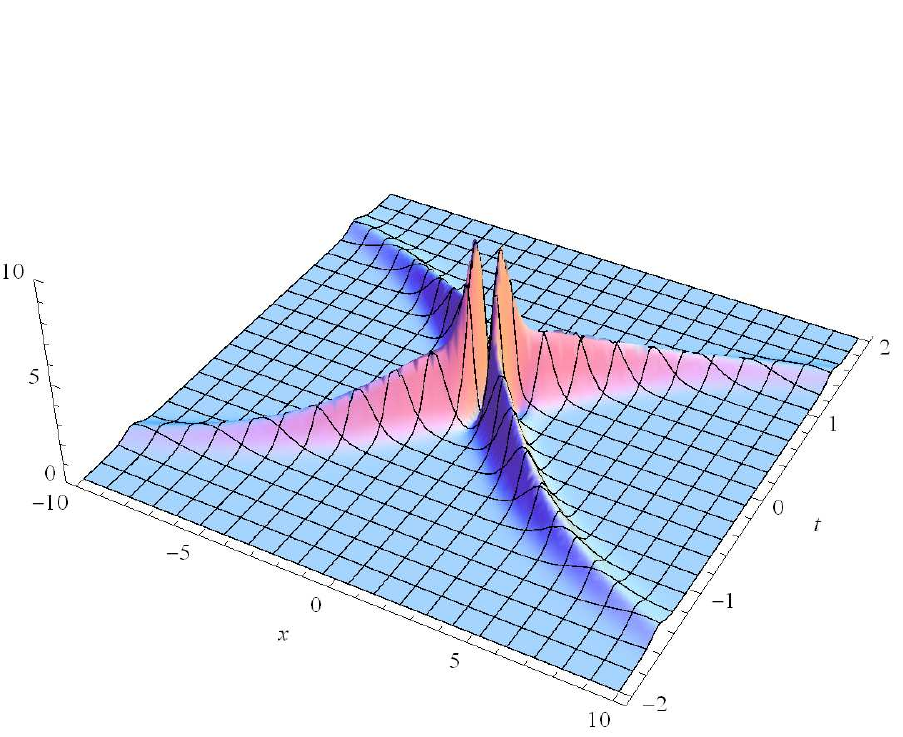}
\end{minipage}
}%
\caption{Shape and motion of 2SS to the nonlocal GP equation \eqref{nonlocal GP} with $\beta=1$.
(a) Solution given by \eqref{non-2SS-b} for $k_1=0.1+i, k_2=-0.1-i$, $\theta^{(0)}_1=\theta^{(0)}_2=0$ and $\nu=0.6$.
(b) Solution given by \eqref{non-2SS-b} for $k_1=1+i, k_2=-1-i$, $\theta^{(0)}_1=\theta^{(0)}_2=0$ and $\nu=0.6$.}
\label{fig-7}
\end{figure}


\subsubsection{Jordan block solution }\label{sec-4-2-3}

For Jordan block solutions of the nonlocal GP equation \eqref{GP00} with $\beta=1$,
we consider $N=2$, in this case,  the double Wronskian entry vectors $\phi$ and $\psi$ can be taken as
\begin{equation}
\phi=\bigl(e^{-\theta_1},\partial_{k_1}e^{-\theta_1},
e^{\theta_1^*(-x)},\partial_{k_1^*}e^{\theta_1^*(-x)}\bigr)^T,
\psi=\bigl(-e^{\theta_1},-(\partial_{k_1^*}e^{\theta_1^*})^*,
e^{-\theta_1^*(-x)},(\partial_{k_1}e^{-\theta_1(-x)})^*\bigr)^T,
\end{equation}
and $\theta_1$ is defined as \eqref{wron-entr-b}.

\section{Conclusion}\label{sec-5}

We have derived the classical and nonlocal GP equation \eqref{GP00} from the second order nonisospectral AKNS system
\eqref{nonAKNS} that allows both classical and nonlocal reductions.
Double Wronskian solutions of the general GP equation \eqref{GP00} are obtained by employing
a reduction technique that has been proved efficient in generating solutions for both classical and nonlocal systems
\cite{ChenDLZ-SPAM-2018,ChenZ-AML-2018,Deng-AMC-2018,Chen-AML-2019,ShiYing-ND-2019}.
We mainly analyzed dynamics of the obtained solutions of the classical GP equation \eqref{GP2}.
It is shown that all these solutions are space-time localized.
In particular, we found that when two solitons travel with same speed, the interaction leads to
and localized oscillating wave carried by a bell-shape soliton,
which is, to our knowledge, not reported before, and maybe catch attention in the experiments of BECs.

\subsection*{Acknowledgments}

This project is  supported by the NSF of China (Nos.11875040, 11631007, 11571225).

\begin{appendix}

\section{Proof of Theorem \ref{Theorem 1}}\label{app-A}

We start from two known lemmas.
\begin{lemma}\label{lemma 2}\cite{FreN-PLA-1983}
Suppose that $\mathbf{D}$ is an arbitrary $s\times (s-2)$ matrix, and $\mathbf{a}$, $\mathbf{b}$,
$\mathbf{c}$ and $\mathbf{d}$ are $s$-order column vectors, then
$$|\mathbf{D}, \mathbf{a},\mathbf{b}||\mathbf{D},\mathbf{c},\mathbf{d}|-|\mathbf{D}, \mathbf{a},\mathbf{c}||\mathbf{D},\mathbf{b},\mathbf{d}|
+|\mathbf{D}, \mathbf{a},\mathbf{d}||\mathbf{D},\mathbf{b}, \mathbf{c}|=0.$$
\end{lemma}

\begin{lemma}\label{lemma 1}\cite{ZDJ-arxiv,ZhaZSZ-RMP-2014}
Suppose that  $\Xi=(a_{js})_{M\times M}$ is an $M\times M$ matrix with
column vector set $\{\alpha_j\}$ and  row vector set $\{\beta_j\}$.
${\cal{P}}=(P_{js})_{M\times M}$ is an $M\times M$ operator matrix where
each $P_{js}$ is an operator. Then we have
\begin{equation}\label{id-D}
\sum^M_{s=1} |\alpha_1, \cdots,\alpha_{s-1},~C_s\alpha_s, ~\alpha_{s+1},\cdots,\alpha_M|
 =\sum^M_{j=1}\left|~\begin{matrix}
 \beta_1\\
 \vdots\\
 \beta_{j-1}\\
 R_j \beta_j\\
 \beta_{j+1}\\
 \vdots\\
 \beta_M
\end{matrix}~\right|,
\end{equation}
where
$C_s\alpha_s=\left(P_{1s} a_{1s}, ~P_{2s} a_{2s},\cdots, P_{Ms} a_{Ms}\right)^T$
and
$R_j\beta_j=\left(P_{j1} a_{j1}, ~P_{j2} a_{j2},\cdots, P_{jM} a_{jM}\right)$.
\end{lemma}

\noindent
\textit{Proof for Theorem \ref{Theorem 1}}:
Making use of structure of double Wronskians and dispersion relation  \eqref{wron-cond-x}, we have
\begin{eqnarray*}
f_x&=& |\W{N-2},N; \W{N-1}|+ |\W{N-1}; \W{N-2},N |,\\
f_{xx}&=& |\W{N-3},N-1,N; \W{N-1}|+|\W{N-2},N+1;\W{N-1}|\\
&&+2|\W{N-2},N; \W{N-2},N|+|\W{N-1}; \W{N-3},N-1,N|+|\W{N-1}; \W{N-2},N+1|,\\
\frac{1}{2}g_x&=&|\W{N-3},N-1; \W{N}|+|\W{N-2}; \W{N-1},N+1|,\\
\frac{1}{2} g_{xx}&=&|\W{N-4},N-2,N-1; \W{N}|+|\W{N-3},N; \W{N}|+2|\W{N-3},N-1;\W{N-1},N+1|\\
&&+|\W{N-2}; \W{N-2},N,N+1|+|\W{N-2}; \W{N-1},N+2|,\\
f_t&=& 2i(|\W{N-3},N-1,N; \W{N-1}|-|\W{N-2},N+1; \W{N-1}| -|\W{N-1}; \W{N-3},N-1,N|\\
&&+|\W{N-1}; \W{N-2},N+1|),\\
\frac{1}{2} g_t &=&2i(|\W{N-4},N-2,N-1; \W{N}|-|\W{N-3},N; \W{N}|-|\W{N-2}; \W{N-2},N,N+1|\\
&&+|\W{N-2}; \W{N-1},N+2|)+\frac{i\delta x^{2}}{2}g.
\end{eqnarray*}

Next, in order to simplify Wronskian verification,
we derive some relations of double Wronskians   using Lemma \ref{lemma 1}.
Taking $\Xi=\DW{N-1}{N-1}$ and for $1\leq j \leq 2N$,
\[ P_{js}=\left\{
\begin{array}{ll}
\frac{i\alpha(t)x}{2}+\partial_x+\frac{i\alpha(t)}{2}s\partial_x^{-1},~ & 1\leq s \leq N,\\
\frac{i\alpha(t)x}{2}-\partial_x+\frac{i\alpha(t)}{2}s\partial_x^{-1},~ & N+1\leq s \leq 2N,
\end{array}
\right.
\]
one can find from \eqref{id-D} that
\begin{equation*}
(\mathrm{Tr}(A))f=|\W{N-2},N; \W{N-1}|-|\W{N-1}; \W{N-2},N|,
\end{equation*}
where $\mathrm{Tr}(A)$ is the trace of $A$.
In a similar way, one has
\begin{equation*}
\begin{array}{rl}
~ & (\mathrm{Tr}(A))(|\W{N-2},N; \W{N-1}|-|\W{N-1}; \W{N-2},N|)\\
= &|\W{N-3},N-1,N; \W{N-1}|+|\W{N-2},N+1; \W{N-1}| -2|\W{N-2},N; \W{N-2},N|\\
~ & +|\W{N-1}; \W{N-3},N-1, N|+|\W{N-1}; \W{N-2},N+1|.
\end{array}
\end{equation*}
It then follows from the  trivial equality
$
f \left\{Tr(A)\left[Tr(A) f \right]\right \}
=\left[Tr(A) f \right]^2,
$
that
\begin{equation}
\begin{array}{rl}
~ & f \left(|\W{N-3},N-1,N; \W{N-1}|+|\W{N-2},N+1;
\W{N-1}|-2|\W{N-2},N; \W{N-2},N|\right.\\
~ & ~~~~\left. +|\W{N-1}; \W{N-3},N-1,N|+|\W{N-1}; \W{N-2},N+1|\right)\\
= & \left( |\W{N-2},N; \W{N-1}|-|\W{N-1}; \W{N-2},N|\right)^2.
\end{array}\label{25}
\end{equation}
In a same manner, we can have
\begin{equation}\label{26}
\begin{array}{rl}
~ & 2f \left(-|\W{N-4},N-2,N-1; \W{N}|+|\W{N-3},N; \W{N}|
      -2|\W{N-3},N-1;\W{N-1},N+1|\right.\\
~ & ~~~~\left. +|\W{N-2}; \W{N-2},N,N+1|+|\W{N-2}; \W{N-1},N+2|\right)-i\alpha(t)gf\\
=& 2\left( |\W{N-2},N; \W{N-1}|-|\W{N-1}; \W{N-2},N|)(|\W{N-3},N-1; \W{N}|-|\W{N-2}; \W{N-1},N+1|\right),
\end{array}
\end{equation}
and
\begin{equation}\label{27}
\begin{array}{rl}
~ & g \left(|\W{N-3},N-1,N; \W{N-1}|+|\W{N-2},N+1;
\W{N-1}|-2|\W{N-2},N; \W{N-2},N|\right.\\
~ & ~~~~\left. +|\W{N-1}; \W{N-3},N-1,N|+|\W{N-1}; \W{N-2},N+1|\right)\\
=& 2\left( |\W{N-2},N; \W{N-1}|-|\W{N-1}; \W{N-2},N|)(|\W{N-3},N-1; \W{N}|-|\W{N-2}; \W{N-1},N+1|\right),
\end{array}
\end{equation}
which are derived respectively from the relations
\begin{equation*}
f \left\{Tr(A)\left[Tr(A) g \right]\right \}
=[Tr(A) f][ Tr(A) g ],
\end{equation*}
\begin{equation*}
g \left\{Tr(A)\left[Tr(A) f \right]\right \}
=[Tr(A) f][ Tr(A) g ].
\end{equation*}

Now, substituting $f, g, f_x, g_x, f_t, g_t, f_{xx}, g_{xx}$ into equation \eqref{22a}, the
left-hand side gives rise to
\begin{equation*}
\begin{array}{rl}
~&(iD_t+D_x^2)g\cdot f\\
=&  4
\biggl[2f \left(|\W{N-2}; \W{N-2},N,N+1| +|\W{N-3},N; \W{N}|\right)\\
~& +g \left(|\W{N-3},N-1,N; \W{N-1}|+|\W{N-1},N; \W{N-2},N+1|\right)\biggr]\\
~&-4\biggl[\left(|\W{N-3},N-1; \W{N}|+|\W{N-2};\W{N-1},N+1|\right)\left(|\W{N-2},N;\W{N-1}|+|\W{N-1}; \W{N-2},N|\right)\\
~&+\left(|\W{N-3},N-1; \W{N}|-|\W{N-2};\W{N-1},N+1|\right)\left(|\W{N-2},N;\W{N-1}|-|\W{N-1}; \W{N-2},N|\right)\biggr]\\
~&-i\alpha(t)gf-\delta x^{2}gf,
\end{array}
\end{equation*}
where we have made use of \eqref{26} and \eqref{27}.
Noting that
\begin{equation*}
2f |\W{N-3},N; \W{N}|+ g|\W{N-3},N-1,N;\W{N-1}|-2|\W{N-3},N-1; \W{N}||\W{N-2},N; \W{N-1}|=0
\end{equation*}
and
\begin{equation*}
2f |\W{N-2}; \W{N-2},N,N+1|+ g|\W{N-1};  \W{N-2},N+1|-2|\W{N-2}; \W{N-1},N+1||\W{N-1};\W{N-2},N|\!=\!0
\end{equation*}
due to Lemma \ref{lemma 2},
we can immediately find that \eqref{22a} is valid.
\eqref{22b} can be verified similarly.
For the equation \eqref{22c},
by substitution of $f, f_x, f_{xx}$ and by virtue of \eqref{25},
one has
\[
D^2_x f\cdot f = 8 f  |\W{N-2},N; \W{N-2},N|-8|\W{N-2},N;\W{N-1}||\W{N-1}; \W{N-2},N|,
\]
the right hand side of which, in light of Lemma \ref{lemma 2},
reduces to $8 |\W{N-2}; \W{N}||\W{N}; \W{N-2}|$, i.e. $-2 gh$.
Thus, \eqref{22c} is proved as well.

\end{appendix}


\begin{thebibliography}{20}

\bibitem{Gross-1961}E.P. Gross,
                    Structure of a quantized vortex in boson systems,
                    Il Nuovo Cimento, 20 (1961) 454-477.


\bibitem{Pitaevskii-1961}L.P. Pitaevskii,
                         Vortex lines in an imperfect Bose gas,
                         Sov. Phys. JETP, 13 (1961) 451-454.


\bibitem{Gross-1963}E.P. Gross,
                    Hydrodynamics of a superfluid condensate,
                    J. Math. Phys., 4 (1963) 195-207.


\bibitem{Liu-2019}W.M. Liu, E. Kengne,
                    Schr\"odinger Equations in Nonlinear Systems,
                    Springer, Singapore, (2019).


\bibitem{Brazhnyi-2003}V.A. Brazhnyi, V.V. Konotop,
                       Evolution of a dark soliton in a parabolic potential: Application to Bose-Einstein condensates,
                       Phys. Rev. A, 68 (2003) No.043613 (10pp).

\bibitem{Serkin-2007}V.N. Serkin, A. Hasegawa, T.L. Belyaeva,
                     Nonautonomous solitons in external potentials,
                     Phys. Rev. Lett., 98 (2007) No.074102 (4pp).


\bibitem{Liang-2005}Z.X. Liang, Z.D. Zhang, W.M. Liu,
         Dynamics of a bright soliton in Bose-Einstein condensates with time-dependent atomic scattering length in an expulsive parabolic potential,
         Phys. Rev. Lett., 94 (2005) No.050402 (4pp).

\bibitem{Zhang-2009}X.F. Zhang, X.H. Hu, X.X. Liu, W.M. Liu,
         Vector solitons in two-component Bose-Einstein condensates with tunable interactions and harmonic potential,
         Phys. Rev. A, 79 (2009) No.033630 (6pp).



\bibitem{Tempere-2001}
         J. Tempere, J.T. Devreese, E.R.L. Abraham,
         Vortices in Bose-Einstein condensates confined in a multiply connected Laguerre-Gaussiaon optical trap,
         Phys. Rev. A, 64 (2001) No.023603 (8pp).

\bibitem{Delande-2014}  D. Delande, K. Sacha,
         Many-body matter-wave dark soliton,
         Phys. Rev. Lett., 112 (2014) No.040402 (5pp).


 \bibitem{Merhasin-2006} I.M. Merhasin, B.A. Malomed, Y.B. Band,
         Partially incoherent gap solitons in Bose-Einstein condensates,
         Phys. Rev. A, 74 (2006) No.033614 (7pp).

\bibitem{Sakaguchi-2009} H. Sakaguchi, B.A. Malomed,
         Solitary vortices and gap solitons in rotating optical lattices,
         Phys. Rev. A, 79 (2009) No.043606 (11pp).

\bibitem{Dong-2013} G.J. Dong, J. Zhu, W.P. Zhang,
         Polaritonic solitons in a Bose-Einstein condensate trapped in a soft optical lattice,
         Phys. Rev. Lett., 110 (2013) No.250401 (6pp).


\bibitem{Gupta-1979}M.R. Gupta,
                    Exact inverse scattering solution of a non-linear evolution equation in a non-uniform medium,
                    Phys. Lett. A, 72 (1979) 420-422.

\bibitem{Zhang-Y-J-2017}Y.J. Zhang, D. Zhao, W.X. Ma,
                        A unified inverse scattering transformation for the local and nonlocal nonautonomous Gross-Pitaevskii equations,
                        J. Math. Phys., 58 (2017) No.013505 (19pp).

\bibitem{Khawaja-U-A-2009}U.A. Khawaja,
                          Integrability of a general Gross-Pitaevskii equation and exact solitonic solutions of a
                          Bose-Einstein condensate in a periodic potential,
                          Phys. Lett. A,  31 (2009) 2710-2716.

\bibitem{Vinoj-2001}M.N. Vinoj, V.C. Kuriakos, K. Porsezian,
                    Optical soliton with damping and frequency chirping in fibre media,
                    Chaos Solitons and Fractals, 12 (2001) 2569-2575.

\bibitem{Zhang-Y-J-2014}Y.J. Zhang, D. Zhao, H.G. Luo,
                        Multi-soliton management by the integrable nonautonomous nonlinear integro-differential Schr\"{o}dinger equation,
                        Ann. Phys., 350 (2014) 112-123.


\bibitem{Sun-2014}W.R. Sun, B. Tian, Y. Jiang and H.L. Zhen,
                  Double-Wronskian solitons and rogue waves for the inhomogeneous nonlinear Schr\"{o}dinger equation in an inhomogeneous plasma,
                  Ann. Phys., 343 (2014) 215-227.


\bibitem{AM-PRL-2013} M.J. Ablowitz, Z.H. Musslimani,
         Integrable nonlocal nonlinear Schr\"odinger equation,
         Phys. Rev. Lett., 110 (2013) No.064105 (5pp).


\bibitem{AM-Nonl-2016} M.J. Ablowitz, Z.H. Musslimani,
         Inverse scattering transform for the integrable nonlocal nonlinear Schr\"odinger equation,
         Nonlinearity, 29 (2016) 915-946.

\bibitem{Ablowitz-SAPM-2016} M.J. Ablowitz, Z.H. Musslimani,
                  Integrable nonlocal nonlinear equations,
                  Stud. Appl. Math., 139 (2016) 7-59.

\bibitem{Fokas-2016}A.S. Fokas,
         Integrable multidimensional versions of the nonlocal nonlinear Schr\"odinger equation,
         Nonlinearity, 29 (2016) 319-324.

\bibitem{YanY-arxiv-2017} B. Yang, J.K. Yang,
        Transformations between nonlocal and local integrable equations,
        Stud. Appl. Math., 140 (2018) 178-201.

\bibitem{Cau-SAPM-2018} V. Caudrelier,
        Interplay between the Inverse Scattering Method and Fokas's unified transform with an application,
        Stud. Appl. Math., 140 (2018) 3-26.

\bibitem{ChenDLZ-SPAM-2018} K. Chen, X. Deng, S.Y. Lou,  D.J. Zhang,
         Solutions of  nonlocal equations reduced from the AKNS hierarchy,
         Stud. Appl. Math., 141 (2018) 113-141.

\bibitem{Zhou-SAPM-} Z.X. Zhou,
        Darboux transformations and global explicit solutions for nonlocal Davey-Stewartson I equation,
        Stud. Appl. Math., 141 (2018) 186-204.


\bibitem{AM-JPA-2019} M.J. Ablowitz, Z.H. Musslimani,
          Integrable nonlocal asymptotic reductions of physically significant nonlinear equations,
          J. Phys. A: Math. Theor., 52 (2019) No.15LT02 (8pp).

\bibitem{Lou-SAPM-2019} S.Y. Lou,
          Prohibitions caused by nonlocality for nonlocal Boussinesq-KdV type systems,
          Stud. Appl. Math., 143 (2019) 123-138.


\bibitem{Zakharov-Shabat-1972} V.E. Zakharov, A.B. Shabat,
        Exact theory of two-dimensional  self-focusing and one-dimensional self-modulation of waves in nonlinear media,
        Sov. Phys. JETP, 34 (1972) 62-69.

\bibitem{AKNS-1973} M.J. Ablowitz, D.J. Kaup, A.C. Newell, H. Segur,
        Nonlinear evolution equations of physical significance,
        Phys. Rev. Lett., 31 (1973) 125-127.


\bibitem{Abdselam-2019} A. Silem, C. Zhang, D.J. Zhang,
        Dynamics of three nonisospectral nonlinear Schr\"{o}dinger equations,
        Chin. Phys. B, 28 (2019) No.020202 (12pp).

\bibitem{Hirota-1974} R. Hirota,
          A new form of B\"acklund transformations and its relation to the inverse scattering problem,
          Prog. Theor. Phys., 52 (1974) 1498-1512.


\bibitem{ChenZ-AML-2018} K. Chen,  D.J. Zhang,
         Solutions of the nonlocal nonlinear Schr\"odinger hierarchy via reduction,
         Appl. Math. Lett., 75 (2018) 82-88.



\bibitem{Nimmo-NLS-1983} J.J.C. Nimmo,
        A bilinear B\"{a}cklund transformation for the nonlinear Schr\"{o}dinger equation,
        Phys. Lett. A, 99 (1983) 279-280.


\bibitem{Zhang-Hietarinta} D.J. Zhang, J. Hietarinta,
        Generalized double-Wronskian solutions to the nonlinear Schr\"{o}dinger equation,
        preprint, 2005.


\bibitem{FreN-PLA-1983} N.C. Freeman, J.J.C. Nimmo,
         Soliton solutions of the KdV and KP equations: the Wronskian technique,
         Phys. Lett. A, 95 (1983) 1-3.


\bibitem{ZDJ-arxiv} D.J. Zhang,
         Notes on solutions in Wronskian form to soliton equations: Korteweg de Vries-type,
          arXiv:nlin.SI/0603008.

\bibitem{ZhaZSZ-RMP-2014} D.J. Zhang, S.L. Zhao, Y.Y. Sun, J. Zhou,
         Solutions to the modified Korteweg-de Vries equation,
          Rev. Math. Phys.,  26 (2014) No.1430006 (42pp).


\bibitem{Deng-AMC-2018} X. Deng, S.Y. Lou, D.J. Zhang,
         Bilinearisation-reduction approach to the nonlocal discrete nonlinear Schr\"odinger equations,
         Appl. Math. Comput., 332 (2018) 477-483.

\bibitem{Chen-AML-2019} K. Chen, S.M. Liu, D.J. Zhang,
         Covariant hodograph transformations between nonlocal short pulse models and AKNS$(-1)$ system,
         Appl. Math. Lett., 88 (2019) 230-236.

\bibitem{ShiYing-ND-2019} Y. Shi, S.F. Shen, S.L. Zhao,
         Solutions and connections of nonlocal derivative nonlinear Schr\"odinger equations,
         Nonlinear Dyn., 95 (2019) 1257-1267.



\end{thebibliography}
\end{document}